\documentclass{lmcs}
\pdfoutput=1
\usepackage[utf8]{inputenc}

\usepackage{lastpage}
\lmcsdoi{22}{2}{25}
\lmcsheading{}{\pageref{LastPage}}{}{}%
{Jun.~27,~2025}{Jun.~02,~2026}{}

\title{Univalent Enriched Categories and the Enriched Rezk Completion}

\author{Niels van der Weide\lmcsorcid{0000-0003-1146-4161}}

\address{Institute for Computing and Information Sciences, Radboud University, Nijmegen, The Netherlands}
\email{nweide@cs.ru.nl}

\keywords{enriched categories, univalent categories, homotopy type theory, univalent foundations, Rezk completion}

\newtheoremstyle{linkthm}
  {6pt}
  {6pt}
  {\itshape}
  {}
  {\bfseries}
  {{\bfseries .}}
  {5pt plus 1pt minus 1pt}%
  {\link{\href{#3}{\thmname{#1} \thmnumber{#2}}}}

\usepackage{amsthm}
\usepackage{amssymb}
\usepackage{mathtools}
\usepackage{url}
\usepackage{xspace}
\usepackage{xifthen}
\usepackage{tikz-cd}
\usepackage[shortcuts]{extdash}
\usepackage{turnstile}
\usepackage{xargs}
\usepackage{bbold}
\usepackage{hyperref}
\usepackage{cleveref}
\usepackage{bussproofs}

\theoremstyle{linkthm}
\newtheorem{propL}[thm]{Proposition}
\newtheorem{lemL}[thm]{Lemma}
\newtheorem{thrm}[thm]{Theorem}

\crefname{lem}{Lemma}{Lemmata}
\crefname{lemL}{Lemma}{Lemmata}
\crefname{prop}{Proposition}{Propositions}

\theoremstyle{theorem}
\newtheorem{problem}[thm]{Problem}

\crefname{problem}{Problem}{Problems}

\theoremstyle{definition}

\newtheorem{axiom}[thm]{Axiom}

\crefname{defi}{Definition}{Definitions}
\crefname{rem}{Remark}{Remarks}
\crefname{exa}{Example}{Examples}

\newcounter{constructionCounter}
\renewcommand{\theconstructionCounter}{\arabic{section}.\arabic{thm}}
\newenvironment{construction}[2]
{\stepcounter{thm}\refstepcounter{constructionCounter}\noindent{\bfseries\link{\href{#1}{Construction \theconstructionCounter}}} {(for Problem~\ref{#2})}{\bfseries .}\kern-1pt}
  {\qed \vspace{6pt}}

\crefname{constructionCounter}{Construction}{Constructions}
\Crefname{constructionCounter}{Construction}{Constructions}

\newcommand{\UniMath}{\href{https://github.com/UniMath/UniMath}{\nolinkurl{UniMath}}\xspace}
\newcommand{\CoqLink}{\href{https://rocq-prover.org}{\nolinkurl{Rocq}}\xspace}

\newcommand{\coqdocbasebaseurl}{https://nmvdw.github.io/EnrichedCats}

\newcommand{\coqdocbaseurl}{\coqdocbasebaseurl /UniMath.}
\newcommand{\urlhash}{\#}
\newcommand{\coqdocurl}[2]{\coqdocbaseurl #1.html\urlhash #2}
\newcommand{\coqdocurlR}[2]{https://nmvdw.github.io/EnrichedCats/rezk/#1.html\urlhash #2}
\newcommand{\nolinkcoqident}[1]{\nolinkurl{#1}}
\makeatletter
\newcommand{\coqident}{\begingroup\@makeother\#\@coqident}
\newcommand{\@coqident}[3][]{%
	\ifthenelse{\isempty{#2}}%
	{\nolinkcoqident{#3}}%
	{\ifthenelse{\isempty{#1}}%
		{\href{\coqdocurl{#2}{#3}}{\nolinkcoqident{#3}}}%
		{\href{\coqdocurl{#2}{#3}}{\nolinkcoqident{#1}}}}%
	\endgroup}

\makeatother

\newcommand{\link}[1]{\textcolor{blue}{#1}}
\newcommand{\conceptDef}[3]{\href{\coqdocurl{#2}{#3}}{\link{\textbf{#1}}}}
\newcommand{\constfont}[1]{\ensuremath{\mathsf{#1}}}
\newcommand{\cat}[1]{\ensuremath{\constfont{#1}}\xspace}

\newcommand{\defeq}{\coloneqq}
\newcommand{\trunc}[1]{\mathopen{}\left\Vert #1\right\Vert \mathclose{}}
\newcommand{\unitt}{\cat{unit}}
\newcommand{\unittel}{\mathbf{\cat{tt}}}
\newcommand{\bool}{\cat{bool}}

\newcommand{\falseB}{\cat{false}}
\newcommand{\IfThenElse}[3]{\cat{if} \> #1 \> \cat{then} \> #2 \> \cat{else} \> #3}
\newcommand{\sigmatype}[3]{\sum(#1 : #2), #3}
\newcommand{\pitype}[3]{\prod(#1 : #2), #3}
\newcommand{\lambdatm}[2]{\lambda #1. #2}

\newcommand{\app}[2]{#1 \> #2}
\newcommand{\apd}[2]{\app{\cat{apd}_{#1}}{#2}}
\newcommand{\Pred}{\cat{P}}
\newcommand{\X}{\cat{X}}
\newcommand{\Y}{\cat{Y}}
\newcommand{\Z}{\cat{Z}}
\newcommand{\UnivU}{\mathcal{U}}
\newcommand{\UnivV}{\mathcal{V}}
\newcommand{\maxU}[2]{#1 \sqcup #2}
\newcommand{\idpath}[1][]{\cat{idpath}_{#1}}

\newcommand{\idpathD}[1][]{\overline{\cat{idpath}}_{#1}}

\newcommand{\id}[1][]{\operatorname{id}_{#1}}
\newcommand{\comp}[2]{#1 \mathbin{\cdot} #2}

\newcommand{\B}{\cat{B}}
\newcommand{\C}{\cat{C}}

\newcommand{\V}{\cat{V}} 
\newcommand{\w}{w}
\newcommand{\x}{x}
\newcommand{\y}{y}
\newcommand{\z}{z}
\newcommand{\f}{f}
\newcommand{\g}{g}
\newcommand{\h}{h}

\newcommand{\xx}{\overline{\x}}
\newcommand{\yy}{\overline{\y}}
\newcommand{\zz}{\overline{\z}}

\newcommand{\pp}{\overline{p}}
\newcommand{\qq}{\overline{q}}
\newcommand{\Set}{\cat{Set}}
\newcommand{\FSub}[1]{\cat{FSub}(#1)}
\newcommand{\Op}[1]{#1^{\cat{op}}}
\newcommand{\Dialg}[2]{\cat{Dialg}(#1, #2)}
\newcommand{\EFunctor}[2]{[ #1 , #2 ]}
\newcommand{\IsoComma}[2]{#1 /_{\cong} #2}
\newcommand{\FKleisli}[1]{\cat{K}(#1)}
\newcommand{\Kleisli}[1]{\cat{Kleisli}(#1)}
\newcommand{\EM}[1]{\cat{EM}(#1)}
\newcommand{\freealg}[1]{\cat{FreeAlg}_{#1}}
\newcommand{\kleislifunctor}[1]{\cat{incl}_{#1}}

\newcommand{\iso}[2]{#1 \cong #2}

\newcommand{\F}{\cat{F}}
\newcommand{\G}{\cat{G}}

\newcommand{\T}{\cat{T}}
\newcommand{\unitM}[1]{\eta_{#1}}
\newcommand{\muM}[1]{\mu_{#1}}

\newcommand{\nt}{\tau}

\newcommand{\munit}[1][]{\mathbb{1}_{#1}}
\newcommand{\mult}[2]{#1 \otimes #2}
\newcommand{\mlunitImpl}[1]{\cat{l}_{#1}}
\newcommand{\mlunit}[1]{\cat{l}}
\newcommand{\mrunitImpl}[1]{\cat{r}_{#1}}
\newcommand{\mrunit}[1]{\cat{r}}
\newcommand{\mlinvunitImpl}[1]{\cat{l}^{-1}_{#1}}
\newcommand{\mlinvunit}[1]{\cat{l}^{-1}}
\newcommand{\mrinvunitImpl}[1]{\cat{r}^{-1}_{#1}}
\newcommand{\mrinvunit}[1]{\cat{r}^{-1}}
\newcommand{\massocImpl}[3]{\cat{a}_{#1, #2, #3}}
\newcommand{\massoc}[3]{\cat{a}}

\newcommand{\minvassoc}[3]{\cat{a}^{-1}}
\newcommand{\msym}[1][]{\cat{s}_{#1}}
\newcommand{\mhom}[2]{#1 \multimap #2}
\newcommand{\mlam}[1]{\lambda(#1)}
\newcommand{\mevalImpl}[2]{\epsilon_{#1, #2}}
\newcommand{\meval}[2]{\epsilon}

\newcommand{\mfununit}[1]{\epsilon_{#1}}
\newcommand{\mfunmult}[1]{\mu_{#1}}
\newcommand{\mfunchange}[2]{\zeta_{#1}(#2)}

\newcommand{\Ec}{\mathcal{E}}
\newcommand{\Ef}{\mathcal{F}}
\newcommand{\Eg}{\mathcal{G}}
\newcommand{\Eh}{\mathcal{H}}

\newcommand{\Ehom}[3]{#1(#2, #3)}
\newcommand{\EidImpl}[2][]{\cat{id}^{\cat{e}}_{#1}(#2)}
\newcommand{\Eid}[2][]{\cat{id}^{\cat{e}}}
\newcommand{\EcompImpl}[4][]{\cat{comp}_{#1}(#2, #3, #4)}
\newcommand{\Ecomp}[4][]{\cat{comp}}
\newcommand{\EFromArr}[1]{\overrightarrow{#1}}
\newcommand{\EToArr}[1]{\overleftarrow{#1}}
\newcommand{\EPrecomp}[1]{#1^{\cat{pre}}}
\newcommand{\EPostcomp}[1]{#1^{\cat{post}}}

\newcommand{\EfunImpl}[3]{#1(#2, #3)}
\newcommand{\Efun}[3]{#1}

\newcommand{\self}[1]{\cat{self}(#1)}
\newcommand{\FSubE}[1]{\cat{FSub}_{\cat{e}}(#1)}
\newcommand{\OpE}[1]{#1^{\cat{op}}}
\newcommand{\DialgE}[2]{\cat{Dialg}_{\cat{e}}(#1, #2)}
\newcommand{\Change}[2]{#1^*(#2)}
\newcommand{\EFunctorE}[2]{\cat{EFunctor}(#1 , #2)}
\newcommand{\ImE}[1]{\cat{Im}(#1)}
\newcommand{\EME}[1]{\cat{EM}_{\cat{e}}(#1)}
\newcommand{\FKleisliE}[1]{\cat{K}_{\cat{e}}(#1)}
\newcommand{\KleisliE}[1]{\cat{Kleisli}_{\cat{e}}(#1)}
\newcommand{\UnitE}[0]{\mathcal{I}}

\newcommand{\Struct}{\mathcal{S}}
\newcommand{\StructOb}[1]{\cat{P}_{#1}}
\newcommand{\StructOnOb}[1]{\cat{p}_{#1}}
\newcommand{\StructMor}[3]{\cat{H}_{(#1, #2)}(#3)}
\newcommand{\StructUnit}{\cat{p}_{\unitt}}
\newcommand{\StructProd}[2]{#1 \times #2}
\newcommand{\StructHom}[2]{#1 \Rightarrow #2}
\newcommand{\StructCat}[1]{\cat{Str}(#1)}

\newcommand{\StructCatSmash}[1]{\cat{Str}_{\wedge}(#1)}
\newcommand{\SmashEqRel}{\sim_{\wedge}}
\newcommand{\SmashMap}[1]{\app{\cat{map}_{\wedge}}{#1}}
\newcommand{\StructBool}[1]{#1_{\bool}}
\newcommand{\SmashStruct}[3]{#1 \wedge_{#3} #2}
\newcommand{\StructSmash}[2]{#1 \wedge #2}
\newcommand{\StructFunSmash}[2]{\cat{hom}_{*}(#1 , #2)}
\newcommand{\StructHomSmash}[2]{#1 \Rightarrow_* #2}
\newcommand{\SmashCurry}{\cat{curry}}
\newcommand{\SmashUncurry}{\cat{uncurry}}

\newcommand{\HomStruct}[2]{\cat{hom}_{(#1, #2)}}
\newcommand{\DCPOStruct}{\cat{DCPO}}
\newcommand{\DCPPOStruct}{\cat{DCPO}_{\bot}}

\newcommand{\PointedSet}{\cat{Set}_*}
\newcommand{\PointedPoset}{\cat{Poset}_{\bot}}

\newcommand{\UnivCat}{\cat{UnivCat}}
\newcommand{\EnrichCat}[1]{\cat{EnrichCat}_{#1}}
\newcommand{\dEnrichCat}[1]{\cat{dEnrichCat}_{#1}}
\newcommand{\Und}[1]{\cat{Und}_{#1}}

\newcommand{\orthogonal}[2]{#1 \mathrel{\bot} #2}
\newcommand{\OFunc}[2]{\cat{O}_{(#1, #2)}}
\newcommand{\Left}{\mathcal{L}}
\newcommand{\Right}{\mathcal{R}}

\newcommand{\RepFun}[1]{\cat{r}_0(#1)}
\newcommand{\RepNat}[1]{\cat{r}_1(#1)}
\newcommand{\Yon}[1]{\cat{y}_{#1}}

\newcommand{\Rezk}[1]{\mathcal{R}(#1)}
\newcommand{\RezkFun}[1][]{\cat{P}_{#1}}

\newcommand{\RezkH}[1]{\mathcal{R}^{\mathcal{H}}(#1)}

\newcommand{\rcl}[1]{\cat{rcl}_{#1}}
\newcommand{\rcleq}[1]{\cat{rcleq}_{#1}}
\newcommand{\re}{\cat{re}}
\newcommand{\rconcat}{\cat{rconcat}}

\newcommand{\rInd}[5]{\cat{ind}_{\mathcal{R}}(#1, #2, #3, #4, #5)}

\newcommand{\onecell}{\rightarrow}
\newcommand{\twocell}{\Rightarrow}
\newcommand{\homC}[3]{\underline{#1}(#2, #3)}
\newcommand{\precomp}[1]{#1^{\cat{pre}}}
\newcommand{\postcomp}[1]{#1^{\cat{post}}}

\newcommand{\whiskerl}{\vartriangleleft}
\newcommand{\whiskerr}{\vartriangleright}

\newcommand{\lassociatorfull}[3]{\alpha_{#1,#2,#3}}
\newcommand{\lunitor}[1]{\lambda}
\newcommand{\linvunitor}[1]{\lambda^{-1}}
\newcommand{\runitor}[1]{\rho}
\newcommand{\rinvunitor}[1]{\rho^{-1}}
\newcommand{\lassociator}[3]{\alpha}
\newcommand{\rassociator}[3]{\alpha^{-1}}

\newcommand{\monadob}[1]{\cat{ob}_{#1}}
\newcommand{\monadendo}[1]{#1}
\newcommand{\monadendofull}[1]{\cat{mor}_{#1}}
\newcommand{\monadunit}[1]{\cat{\eta}_{#1}}
\newcommand{\monadmult}[1]{\cat{\mu}_{#1}}

\newcommand{\klob}[1]{\cat{ob}_{#1}}
\newcommand{\klmor}[1]{\cat{mor}_{#1}}
\newcommand{\klcell}[1]{\cat{cell}_{#1}}

\newcommand{\klumpmor}[1]{\cat{Kl}_{\cat{mor}}(#1)}
\newcommand{\klumpcom}[1]{\cat{Kl}_{\cat{com}}(#1)}
\newcommand{\klumpcell}[1]{\cat{Kl}_{\cat{cell}}(#1)}

\newenvironment{bprooftree}
{\leavevmode\hbox\bgroup}
{\DisplayProof\egroup}

\begin{document}

\maketitle

\begin{abstract}
Enriched categories are categories whose sets of morphisms are enriched with extra structure.
Such categories play a prominent role in the study of higher categories, homotopy theory, and the semantics of programming languages.
In this paper, we study univalent enriched categories.
We prove that all essentially surjective and fully faithful functors between univalent enriched categories are equivalences,
and we show that every enriched category admits a Rezk completion.
Finally, we use the Rezk completion for enriched categories to construct univalent enriched Kleisli categories.
\end{abstract}

\section{Introduction}
\label{sec:intro}
Over the years, category theory \cite{mac2013categories} has established itself as a powerful mathematical framework with a wide variety of applications.
The applications of category theory range from pure mathematics \cite{MR2840650,MR1269324} to computer science \cite{Moggi89,plotkin2002notions,plotkin2003algebraic,Power00}.
This study resulted in the development of various notions of categories.

One of these notions is given by \emph{enriched categories}.
Enriched categories are categories whose morphisms are equipped with additional structure.
Examples of such categories are plentiful.
For instance, in the study of the semantics of effectful programming languages,
one uses categories enriched over \emph{directed complete partial orders} (DCPOs) \cite{plotkin2002notions,plotkin2003algebraic,Power00}.
The type of morphisms in categories enriched over DCPOs is given by a DCPO,
and thus fixpoint equations of morphisms can be solved in such categories \cite{Wand79}.
For similar purposes, categories enriched over partial orders have been used \cite{McDermottM22}.
Other applications of enriched categories appear in homological algebra \cite{MR1269324}
where one is interested in categories enriched over abelian groups,
abstract homotopy theory \cite{MR2840650}
where one looks at categories enriched over simplicial sets,
and higher category theory \cite{MR1883478}
where one considers categories enriched over categories.

\subsection*{Univalent Foundations.}
Throughout this paper, we work in univalent foundations \cite{rijke2022introduction,hottbook}.
Univalent foundations is an extension of dependent Martin-L\"of Type Theory \cite{MLTT} with the \emph{univalence axiom}.
This axiom says that the identity of types is the same as equivalences between them.
More specifically, we have a map that sends identities $A = B$ to equivalences $A \simeq B$,
and the univalence axiom states that this map is an equivalence.
Concretely, this means that properties of types are invariant under equivalence
and that two types share the same properties whenever we have an equivalence between them.

Univalent foundations is especially interesting for the study of category theory.
In category theory, objects are generally viewed up to isomorphism:
whenever there is an isomorphism between two objects, they share the same categorical properties.
This is known as the \emph{principle of equivalence},
and this principle is made precise using \emph{univalent categories}.

\subsection*{Univalent Categories.}
In univalent foundations, the ``correct'' notion of category is given by univalent categories.
Given two objects $\x$ and $\y$ in a category,
we have a map sending identities $\x = \y$ to isomorphisms $\iso{\x}{\y}$.
In a univalent category, this mapping is required to be an equivalence:
identities between objects are thus the same as isomorphisms.
Hence, whenever two objects are isomorphic, they satisfy the same properties.
Semantically, this is the correct notion of category,
because in the simplicial set model, univalent categories correspond to set-theoretic categories \cite{simpset}.

There are several consequences of univalence for categories.
For instance, we have a \emph{structure identity principle} for univalent categories.
This principle says that the identity type of two categories is equivalent to the type of adjoint equivalences between them \cite{rezk_completion}.
As a consequence, one gets that whenever two categories are equivalent, then they have the same properties.
Another consequence is that every essentially surjective fully faithful functor 
is an adjoint equivalence as well.
Usually, one uses the axiom of choice to prove this principle,
but if the domain is univalent, then one can constructively prove this fact.
Finally, every category is weakly equivalent to a univalent one, called its \emph{Rezk completion} \cite{rezk_completion}.

While most categories that one encounters in practice are univalent (e.g., Eilenberg-Moore categories and functor categories),
some are not.
An example is given by the Kleisli category.
Usually, the Kleisli category $\FKleisli{\T}$ of a monad $\T$ on a category $\C$ is defined to be the category
whose objects are objects of $\C$
and whose morphisms from $\x$ to $\y$ are morphisms $\x \onecell \app{\T}{\y}$ in $\C$.
However, this does not give rise to a univalent category in general.
One can give an alternative presentation of the Kleisli category as a full subcategory $\Kleisli{\T}$ of the Eilenberg-Moore category to obtain a univalent category \cite{univalence-principle}.
To prove the desired theorems about $\Kleisli{\T}$, one uses that it is the Rezk completion of $\FKleisli{\T}$ \cite{Weide23}.

\subsection*{Univalent Enriched Categories.}
In this paper, we develop enriched category theory in univalent foundations.
More specifically, we define univalent enriched categories,
and we prove analogous theorems for univalent enriched categories as for univalent categories.
We show that univalent enriched categories satisfy a structure identity principle,
that every essentially surjective fully faithful functor is an adjoint equivalence,
and that every enriched category admits a Rezk completion.
We also use these theorems to construct univalent enriched Kleisli categories.

\subsection*{Related work.}
While there are numerous libraries that contain a formalization of categories,
enriched categories have gotten less attention.
Several libraries,
such as Agda categories \cite{HuC21} using Agda \cite{norell2009dependently},
mathlib \cite{X20} using Lean \cite{demoura:2015},
and the category-theory library \cite{coq:cat-th} in \CoqLink \cite{Coq:manual},
contain a couple of basic definitions.
In \UniMath, Satoshi Kura also formalized several basic concepts of enriched categories.
However, none of the aforementioned formalizations consider much of the theory of enriched categories,
and they do not consider univalent enriched categories.
In addition, we use enrichments (\cref{def:enrichment}), while the other formalizations use the definition as given by Kelly \cite{kelly1982basic}.
Enrichments have been used in the setting of skew-enriched categories \cite{MR3800887},
and, with a slightly different definition, in the study of strong monads \cite[Definition 5.1]{abs-2207-00851}.

\subsection*{Formalization.}
The results in this paper are formalized in the \CoqLink \cite{Coq:manual} proof assistant using the \UniMath library \cite{UniMath}.
We use the \UniMath library in this work,
because we frequently use notions from bicategory theory that have only been formalized in UniMath up to now.
\Cref{constr:enriched-rezk-completion-HIT} is not integrated in the \UniMath repository, but it is available in another repository\footnote{\url{https://github.com/nmvdw/RezkCompletion}}.
Definitions and theorems are accompanied with links to their corresponding identifier in the formalization,
and these links are \link{blue}.
The tool \texttt{coqwc} reports the following number of lines of code in the formalization.
\begin{verbatim}
     spec    proof comments
    20812     9237      558 total
\end{verbatim}
Besides what is discussed in this paper,
the formalization also contains (weighted) limits and colimits in enriched categories
and models of the enriched effect calculus \cite{EggerMS14}.

\subsection*{Version History.}
This paper is an extended version of the conference paper ``Univalent Enriched Categories and the Enriched Rezk Completion'' published at FSCD~\cite{vanderweide:2024b}.
Compared to that version the following changes have been made.

\begin{itemize}
  \item In \Cref{def:enrichment}, some laws were missing, although they were present in the formalization.
    Now all necessary laws are present in \Cref{def:enrichment}.
  \item We added sketches for the proofs of \Cref{prop:inv2cell-enriched} and \Cref{thm:univalence-principle-enriched-cats}.
  \item Two examples are added in \Cref{sec:examples},
    namely the unit (\Cref{exa:unit-enriched}) and the tensor of enriched categories (\Cref{exa:tensor-enriched}).
    In addition, we explain that these enriched categories are not necessarily univalent.
  \item \Cref{def:structure} has been expanded to include Cartesian closed structures.
  \item We added \Cref{sec:smash-products}, in which we define a notion of structure that support the smash product,
    and we characterize enrichments over such structures.
  \item \Cref{sec:enriched-rezk} has been rewritten to include a construction of the Rezk completion using higher inductive types (\cref{constr:enriched-rezk-completion-HIT})
    and an explanation of the universal property of the Rezk completion.
    In addition, we added a discussion about enriched profunctors,
    and we argue why to construct the monoidal bicategory of enriched categories,
    one needs to use higher inductive types.
\end{itemize}

\subsection*{Contributions and Overview.}
The contributions of this paper are as follows.
\begin{itemize}
  \item A construction of the bicategory of univalent enriched categories (\cref{def:disp-bicat-enrichment})
    and a proof that this bicategory is univalent (\cref{thm:univalence-principle-enriched-cats});
  \item a construction of the image factorization system of enriched categories (\cref{constr:image-factorization});
  \item a proof that all fully faithful and essentially surjective enriched functors are adjoint equivalences (\cref{thm:enriched-factorization});
  \item two constructions of the Rezk completion for enriched categories (\cref{constr:enriched-rezk-completion,constr:enriched-rezk-completion-HIT})
    and a proof of their universal property (\cref{thm:rezk-completion-ump});
  \item a construction of Kleisli objects (\cref{constr:enriched-kleisli-cat}) in the bicategory of univalent enriched categories.
\end{itemize}
In \cref{sec:examples}, we discuss numerous examples of enriched categories.

\section{The Bicategory of Enriched Categories}
\label{sec:bicat-enriched-cat}
In the remainder of this paper, we study univalent enriched categories,
and in this section, we discuss a structure identity principle for univalent enriched categories,
which says that identity of univalent enriched categories is the same as equivalence.
Before we do so, we briefly recall \emph{monoidal categories} to fix the notation for the remainder of the paper \cite{DBLP:conf/cpp/AhrensMWW24,mac2013categories}.

\begin{defi}
\label{def:moncat}
A \conceptDef{monoidal category}{CategoryTheory.Monoidal.Categories}{monoidal_cat} consists of a category $\V$ together with
\begin{itemize}
  \item an object $\munit[\V] : \V$;
  \item a bifunctor $\mult{-}{-} : \V \times \V \onecell \V$;
  \item isomorphisms $\mlunitImpl{\x} : \mult{\munit[\V]}{\x} \onecell \x$, $\mrunitImpl{\x} : \mult{\x}{\munit[V]} \onecell \x$,
    and $\massocImpl{\x}{\y}{\z} : \mult{(\mult{\x}{\y})}{\z} \onecell \mult{\x}{(\mult{\y}{\z})}$;
\end{itemize}
such that $\mlunitImpl{-}$, $\mrunitImpl{-}$, and $\massocImpl{-}{-}{-}$ are natural and such that the triangle and pentagon laws hold.

A \conceptDef{symmetric monoidal category}{CategoryTheory.Monoidal.Structure.Symmetric}{sym_monoidal_cat} is a monoidal category $\V$ together with morphisms $\msym[\x,\y] : \mult{\x}{\y} \rightarrow \mult{\y}{\x}$,
such that $\msym[\x,\y]$ and $\msym[\y,\x]$ are inverses, $\msym[-,-]$ is natural, and satisfies the hexagon law.

A \conceptDef{symmetric monoidal closed category}{CategoryTheory.Monoidal.Structure.Closed}{sym_mon_closed_cat} is a symmetric monoidal category $\V$
such that for every $\x : \V$ the functor $\mult{x}{-}$ has a right adjoint.
\end{defi}

If we have a symmetric monoidal closed category,
then we have internal homs $\mhom{\x}{\y}$.
We also have evaluation morphisms $\mevalImpl{\x}{\y} : \mult{(\mhom{\x}{\y})}{\x} \onecell \y$,
and an internal lambda abstraction operation $\mlam{\f} : x \onecell \mhom{y}{z}$ for every $\f : \mult{\x}{\y} \onecell z$.

Usually, a $\V$-enriched category consists of a collection of objects together with a hom-object $\Ehom{\Ec}{\x}{\y}$ in $\V$ for all $\x$ and $\y$,
such that we have appropriate identity morphisms and a composition operation \cite{kelly1982basic}.
Every enriched category $\Ec$ has an underlying category, which has the same collection of objects
and whose morphisms from $\x$ to $\y$ are the same as morphisms $\munit \onecell \Ehom{\Ec}{\x}{\y}$ in $\V$.

However, we take a slightly different approach: we use \emph{enrichments}.
An enrichment for a category $\C$ consists of a hom-object $\Ehom{\Ec}{\x}{\y}$ in $\V$ for all $\x$ and $\y$,
such that we have the appropriate identity and composition morphisms
and such that $\C$ is equivalent to the underlying category of the corresponding enriched category.
As such, we view an enriched category as a category together with extra structure.
This viewpoint also determines our notion of univalence for enriched categories:
a univalent enriched category is an enriched category such that its underlying category is univalent.
Note that our notion of univalence is similar to completeness for enriched $\infty$-categories \cite{MR3345192}.
Using enrichments, we can equivalently phrase univalent enriched categories as a univalent category together with an enrichment.

\begin{defi}
\label{def:enrichment}
Let $\V$ be a monoidal category and let $\C$ be a category.
A \conceptDef{$\V$-enrichment}{CategoryTheory.EnrichedCats.Enrichment}{enrichment} $\Ec$ for $\C$ consists of
\begin{itemize}
  \item for all objects $\x, \y : \C$ an object $\Ehom{\Ec}{\x}{\y} : \V$;
  \item for every object $\x : \C$ a morphism $\EidImpl[\Ec]{\x} : \munit[V] \onecell \Ehom{\Ec}{\x}{\x}$;
  \item for all objects $\x, \y, \z : \C$ a morphism $\EcompImpl[\Ec]{\x}{\y}{\z} : \mult{\Ehom{\Ec}{\y}{\z}}{\Ehom{\Ec}{\x}{\y}} \onecell \Ehom{\Ec}{\x}{\z}$;
  \item for all morphisms $\f : \x \onecell \y$ in $\C$ a morphism $\EFromArr{\f} : \munit[V] \onecell \Ehom{\Ec}{\x}{\y}$;
  \item for every morphism $\f : \munit[V] \onecell \Ehom{\Ec}{\x}{\y}$ a morphism $\EToArr{\f} : \x \onecell \y$ in $\C$.
\end{itemize}
In addition, we require that $\EFromArr{\EToArr{\f}} = \f$ and $\EToArr{\EFromArr{\f}} = \f$,
and that the following diagrams commute.
\[
  \begin{tikzcd}[column sep = 4em]
    {\mult{\munit}{\Ehom{\Ec}{\x}{\y}}} & {\mult{\Ehom{\Ec}{\y}{\y}}{\Ehom{\Ec}{\x}{\y}}} \\
    & {\Ehom{\Ec}{\x}{\y}}
    \arrow["{\mult{\EidImpl{\y}}{\id}}", from=1-1, to=1-2]
    \arrow["{\EcompImpl{\x}{\y}{\y}}", from=1-2, to=2-2]
    \arrow["{\mlunitImpl{\Ehom{\Ec}{\x}{\y}}}"', from=1-1, to=2-2]
  \end{tikzcd}
  \quad \quad
  \begin{tikzcd}[column sep = 4em]
    {\mult{\Ehom{\Ec}{\x}{\y}}{\munit}} & {\mult{\Ehom{\Ec}{\x}{\y}}{\Ehom{\Ec}{\x}{\x}}} \\
    & {\Ehom{\Ec}{\x}{\y}}
    \arrow["{\mult{\id{}}{\EidImpl{\x}}}", from=1-1, to=1-2]
    \arrow["{\EcompImpl{\x}{\x}{\y}}", from=1-2, to=2-2]
    \arrow["{\mrunitImpl{\Ehom{\Ec}{\x}{\y}}}"', from=1-1, to=2-2]
  \end{tikzcd}
\]
\[
  \begin{tikzcd}
    {\mult{(\mult{\Ehom{\Ec}{\y}{\z}}{\Ehom{\Ec}{\x}{\y}})}{\Ehom{\Ec}{\w}{\x}}} & {\mult{\Ehom{\Ec}{\y}{\z}}{(\mult{\Ehom{\Ec}{\x}{\y}}{\Ehom{\Ec}{\w}{\x}})}} \\
    & {\mult{\Ehom{\Ec}{\y}{\z}}{\Ehom{\Ec}{\w}{\y}}} \\
    {\mult{\Ehom{\Ec}{\x}{\z}}{\Ehom{\Ec}{\w}{\x}}} & {\Ehom{\Ec}{\w}{\z}}
    \arrow["{\massoc{}{}{}}", from=1-1, to=1-2]
    \arrow["{\mult{\id}{\EcompImpl{\w}{\x}{\y}}}", from=1-2, to=2-2]
    \arrow["{\EcompImpl{\w}{\y}{\z}}", from=2-2, to=3-2]
    \arrow["{\mult{\EcompImpl{\x}{\y}{\z}}{\id}{}}"', from=1-1, to=3-1]
    \arrow["{\EcompImpl{\w}{\x}{\z}}"', from=3-1, to=3-2]
  \end{tikzcd}
\]
Finally,
we require that $\EToArr{\EidImpl{\x}} = \id{x}$
and that $\f \cdot \g = \EToArr{\h}$
for all $\f : \x \rightarrow \y$ and $\g : \y \onecell \z$
where $\h$ is the following composition of morphisms.
\[
\begin{tikzcd}[column sep = 4em]
  {\munit[V]} & {\mult{\munit[V]}{\munit[V]}} & {\mult{\Ehom{\Ec}{\y}{\y}}{\Ehom{\Ec}{\x}{\y}}} & {\Ehom{\Ec}{\x}{\z}}
  \arrow["\mlinvunit{}", from=1-1, to=1-2]
  \arrow["{\mult{\EFromArr{\f}}{\EFromArr{\g}}}", from=1-2, to=1-3]
  \arrow["\EcompImpl{\x}{\y}{\z}", from=1-3, to=1-4]
\end{tikzcd}
\]
\end{defi}

When it is clear from the context, we leave the arguments of $\Ecomp{}{}{}$ and $\Eid{}$ implicit.
In addition, note that the morphism $\Eid{}$ is redundant in \cref{def:enrichment},
since we have that $\EidImpl[\Ec]{\x} = \EFromArr{\id}$.
However, we decided to keep $\Eid{}$ in the definition,
because then it is slightly more convenient to prove \cref{prop:enrichment-equiv}.

In the remainder, we use the following operations for $\V$-enrichments $\Ec$ for a category $\C$.
Given an object $\w$ and a morphism $\f : \x \onecell \y$, we define $\EPrecomp{\f}$ as the following composition.
\[
\begin{tikzcd}[column sep = 4em]
  {\Ehom{\Ec}{\w}{\x}} & {\mult{\munit}{\Ehom{\Ec}{\w}{\x}}} & {\mult{\Ehom{\Ec}{\x}{\y}}{\Ehom{\Ec}{\w}{\x}}} & {\Ehom{\Ec}{\w}{\y}}
  \arrow["{\mlinvunitImpl{}}", from=1-1, to=1-2]
  \arrow["{\mult{\EFromArr{\f}}{\id{}}}", from=1-2, to=1-3]
  \arrow["{\Ecomp{\w}{\x}{\y}}", from=1-3, to=1-4]
\end{tikzcd}
\]
For objects $\z$ and morphisms $\f : \x \onecell \y$, we define $\EPostcomp{\f}$ as the following composition.
\[
\begin{tikzcd}[column sep = 4em]
  {\Ehom{\Ec}{\y}{\z}} & {\mult{\Ehom{\Ec}{\y}{\z}}{\munit}} & {\mult{\Ehom{\Ec}{\y}{\z}}{\Ehom{\Ec}{\x}{\y}}} & {\Ehom{\Ec}{\x}{\z}}
  \arrow["{\mrinvunitImpl{}}", from=1-1, to=1-2]
  \arrow["{\mult{\id}{\EFromArr{\f}}}", from=1-2, to=1-3]
  \arrow["{\Ecomp{\x}{\y}{\z}}", from=1-3, to=1-4]
\end{tikzcd}
\]

Note that a category together with a $\V$-enrichment is the same as an enriched category as defined by Kelly \cite{kelly1982basic}.

\begin{propL}[\coqdocurl{CategoryTheory.EnrichedCats.Enriched.EnrichmentEquiv}{enriched_precat_weq_cat_with_enrichment}]
\label{prop:enrichment-equiv}
For a monoidal category $\V$,
the type of categories together with a $\V$-enrichment
is equivalent to the type of $\V$-enriched categories.
\end{propL}
  
The reason why we use enrichments over the usual definition,
is because it simplifies the proof of the structure identity principle for enriched categories.
A structure identity principle is already present for univalent categories \cite[Theorem 6.17]{rezk_completion},
which can be reused directly if we define enriched categories as pairs of univalent categories together with an enrichment.
However, if we would use the definition in \cite{kelly1982basic} instead,
then reusing this principle would be more cumbersome
and thus the desired proof would be more involved.

To phrase the structure identity principle for enriched categories,
we use \emph{univalent bicategories}.
More specifically, this principle for enriched categories is expressed by saying that the bicategory of enriched categories is univalent.
We define the bicategory of enriched categories using \emph{displayed bicategories} \cite{bicatspaper}.
A displayed bicategory over a bicategory $\B$ represents structures and properties to be added to objects, 1-cells, and 2-cells in $\B$.
In our case, we define a displayed bicategory $\dEnrichCat{\V}$ over the bicategory $\UnivCat$ of univalent categories,
and then $\EnrichCat{\V}$ is total bicategory of $\dEnrichCat{\V}$.
The displayed objects over a univalent category $\C$ are $\V$-enrichments for $\C$,
the displayed 1-cells over a functor are enrichments for functors,
and the displayed 2-cells over a natural transformation are proofs that this transformation is enriched.

Note that from the machinery of displayed bicategories,
we get a pseudofunctor $\Und{\V} : \EnrichCat{\V} \onecell \UnivCat$,
which sends every enriched category to its underlying category.
Using this pseudofunctor, we can understand an enrichment for $\C$ to be an object in the fiber of $\C$ along $\Und{\V}$.

\begin{defi}
\label{def:functor-enrichment}
Suppose that we have $\V$-enrichments $\Ec_1$ and $\Ec_2$ for $\C_1$ and $\C_2$ respectively.  
A \conceptDef{$\V$-enrichment}{CategoryTheory.EnrichedCats.EnrichmentFunctor}{functor_enrichment} $\Ef$ for a functor $\F : \C_1 \onecell \C_2$ from $\Ec_1$ to $\Ec_2$
is a family of morphisms $\EfunImpl{\Ef}{\x}{\y} : \Ehom{\Ec_1}{\x}{\y} \onecell \Ehom{\Ec_2}{\app{\F}{\x}}{\app{\F}{\y}}$ such that the following diagrams commute.
\[
  \begin{tikzcd}
    \munit & {\Ehom{\Ec_1}{\x}{\x}} \\
    & {\Ehom{\Ec_2}{\app{\F}{\x}}{\app{\F}{\x}}}
    \arrow["{\Eid{\x}}", from=1-1, to=1-2]
    \arrow["{\EfunImpl{\Ef}{\x}{\x}}", from=1-2, to=2-2]
    \arrow["{\Eid{\app{\F}{\x}}}"', from=1-1, to=2-2]
  \end{tikzcd}
  \quad \quad
  \begin{tikzcd}[column sep = 4em]
    {\mult{\Ehom{\Ec_1}{\y}{\z}}{\Ehom{\Ec_1}{\x}{\y}}} & {\Ehom{\Ec_1}{\x}{\z}} \\
    {\mult{\Ehom{\Ec_2}{\app{\F}{\y}}{\app{\F}{\z}}}{\Ehom{\Ec_2}{\app{\F}{\x}}{\app{\F}{\y}}}} & {\Ehom{\Ec_2}{\app{\F}{\x}}{\app{\F}{\z}}}
    \arrow["{\Ecomp{\x}{\y}{\z}}"', from=2-1, to=2-2]
    \arrow["{\EfunImpl{\Ef}{\x}{\z}}", from=1-2, to=2-2]
    \arrow["{\mult{\EfunImpl{\Ef}{y}{z}}{\EfunImpl{\Ef}{x}{y}}}"', from=1-1, to=2-1]
    \arrow["{\Ecomp{\app{\F}{\x}}{\app{\F}{\y}}{\app{\F}{\z}}}", from=1-1, to=1-2]
  \end{tikzcd}
\]
In addition, we require that $\EFromArr{\app{\F}{\f}} = \EFromArr{\f} \cdot \EfunImpl{\Ef}{\x}{\y}$.
\end{defi}

\begin{defi}
\label{def:nat-trans-enrichment}
Let $\Ef_1$ and $\Ef_2$ be $\V$-enrichments for functors $\F_1, \F_2 : \C_1 \rightarrow \C_2$ from $\Ec_1$ to $\Ec_2$.
A natural transformation $\nt : F_1 \twocell F_2$ is called \conceptDef{$\V$-enriched}{CategoryTheory.EnrichedCats.EnrichmentTransformation}{nat_trans_enrichment} whenever the following diagram commutes.
\[
  \begin{tikzcd}
    &[-5pt] {\mult{\Ehom{\Ec_1}{\x}{\y}}{\munit}} &[10pt] {\mult{\Ehom{\Ec_2}{\app{\F_2}{\x}}{\app{\F_2}{\y}}}{\Ehom{\Ec_2}{\app{\F_1}{\x}}{\app{\F_2}{\x}}}} &[-10pt] \\
    {\Ehom{\Ec_1}{\x}{\y}} &&& {\Ehom{\Ec_2}{\app{\F_1}{\x}}{\app{\F_2}{\y}}} \\
    & {\mult{\munit}{\Ehom{\Ec_1}{\x}{\y}}} & {\mult{\Ehom{\Ec_2}{\app{\F_1}{\y}}{\app{\F_2}{\y}}}{\Ehom{\Ec_2}{\app{\F_1}{\x}}{\app{\F_1}{\y}}}}
    \arrow["{\mrinvunitImpl{}}", from=2-1, to=1-2]
    \arrow["{\mult{\Efun{\Ef_2}{\x}{\y}}{\EFromArr{\app{\nt}{\x}}}}", from=1-2, to=1-3]
    \arrow["{\Ecomp{\app{\F_1}{\x}}{\app{\F_2}{\x}}{\app{\F_2}{\y}}}", from=1-3, to=2-4]
    \arrow["{\mlinvunitImpl{}}"', from=2-1, to=3-2]
    \arrow["{\mult{\EFromArr{\app{\nt}{\y}}}{\Efun{\Ef_1}{\x}{\y}}}"', from=3-2, to=3-3]
    \arrow["{\Ecomp{\app{\F_1}{\x}}{\app{\F_1}{\y}}{\app{\F_2}{\y}}}"', from=3-3, to=2-4]
  \end{tikzcd}
\]
\end{defi}

Note that the condition for $\V$-enriched natural transformations can equivalently be formulated by saying that the following diagram commutes.
\begin{equation}
\label{eq:nat-trans-enrichment}
\begin{tikzcd}[column sep = 4em]
  {\Ehom{\Ec_1}{\x}{\y}} & {\Ehom{\Ec_2}{\app{\F_1}{\x}}{\app{\F_1}{\y}}} \\
  {\Ehom{\Ec_2}{\app{\F_2}{\x}}{\app{\F_2}{\y}}} & {\Ehom{\Ec_2}{\app{\F_1}{\x}}{\app{\F_2}{\y}}}
  \arrow["{\Efun{\Ef_1}{\x}{\y}}", from=1-1, to=1-2]
  \arrow["{\Efun{\Ef_2}{\x}{\y}}"', from=1-1, to=2-1]
  \arrow["{\EPrecomp{(\app{\nt}{\x})}}"', from=2-1, to=2-2]
  \arrow["{\EPostcomp{(\app{\nt}{\y})}}", from=1-2, to=2-2]
\end{tikzcd}
\end{equation}
Now we have everything in place to define the bicategory of enriched categories.

\begin{defi}
\label{def:disp-bicat-enrichment} 
Let $\V$ be a monoidal category.
We define the \conceptDef{displayed bicategory $\dEnrichCat{\V}$ of enrichments}{Bicategories.DisplayedBicats.Examples.EnrichedCats}{disp_bicat_of_enriched_cats} over $\UnivCat$ as follows.
\begin{itemize}
  \item The displayed objects over a category $\C$ are $\V$-enrichments for $\C$;
  \item the displayed 1-cells over a functor $\F : \C_1 \onecell \C_2$ from $\Ec_1$ to $\Ec_2$ are $\V$-enrichments for $\F$ from $\C_1$ to $\C_2$;
  \item the displayed 2-cells over a natural transformation $\nt : \F_1 \twocell \F_2$ from $\Ef_1$ to $\Ef_2$ are proofs that $\nt$ is $\V$-enriched.
\end{itemize}
The \conceptDef{bicategory of enriched categories}{Bicategories.DisplayedBicats.Examples.EnrichedCats}{bicat_of_enriched_cats} is defined to be the total bicategory of $\dEnrichCat{\V}$,
and we denote it by $\EnrichCat{\V}$.
Its objects are \textbf{univalent $\V$-enriched categories}, and we call the 1-cells and 2-cells of $\EnrichCat{\V}$ \textbf{enriched functors} and \textbf{enriched transformations} respectively.
\end{defi}

Note that by construction a univalent $\V$-enriched category is the same as a univalent category together with a $\V$-enrichment.
Objects in univalent $\V$-enriched categories are thus identified up to isomorphism in the underlying category.
In addition, our univalence condition for $\V$-enriched categories has no local variant in contrast to univalence for bicategories \cite{bicatspaper},
since we only look at enrichments over monoidal 1-categories.

To show that \cref{def:disp-bicat-enrichment} actually gives rise to a displayed bicategory,
one also needs to construct enrichments for the identity and composition,
and one needs to prove that the identity transformation is enriched
and that enriched transformations are preserved under composition and whiskering.
Details on this construction are left to the formalization.

\begin{propL}[\coqdocurl{Bicategories.DisplayedBicats.Examples.EnrichedCats}{make_is_invertible_2cell_enriched}]
\label{prop:inv2cell-enriched}
Let $\nt : \Ef_1 \twocell \Ef_2$ be a 2-cell in $\EnrichCat{\V}$.
Then $\nt$ is invertible if the underlying natural transformation of $\nt$ is a natural isomorphism.
\end{propL}

\begin{proof}
Suppose that we have enriched functors $\Ef_1, \Ef_2 : \Ec_1 \onecell \Ec_2$
and an enriched natural transformation $\nt : \Ef_1 \twocell \Ef_2$
whose underlying natural transformation is a pointwise isomorphism.
We need to construct another enriched natural transformation $\nt^{-1} : \Ef_2 \twocell \Ef_1$
such that $\nt$ and $\nt^{-1}$ compose to the identity.
Given $x : \Ec_1$,
we define $\nt^{-1}(x)$ to be the inverse of $\nt(x)$.
To prove naturality,
it suffices to show that the following diagram commutes.
\[
\begin{tikzcd}[column sep = 4em]
  {\Ehom{\Ec_1}{\x}{\y}} & {\Ehom{\Ec_2}{\app{\F_2}{\x}}{\app{\F_2}{\y}}} \\
  {\Ehom{\Ec_2}{\app{\F_1}{\x}}{\app{\F_1}{\y}}} & {\Ehom{\Ec_2}{\app{\F_2}{\x}}{\app{\F_1}{\y}}}
  \arrow["{\Efun{\Ef_2}{\x}{\y}}", from=1-1, to=1-2]
  \arrow["{\Efun{\Ef_1}{\x}{\y}}"', from=1-1, to=2-1]
  \arrow["{\EPrecomp{(\app{\nt^{-1}}{\x})}}"', from=2-1, to=2-2]
  \arrow["{\EPostcomp{(\app{\nt^{-1}}{\y})}}", from=1-2, to=2-2]
\end{tikzcd}
\]
This holds because $\nt$ is an enriched transformation,
and because $\EPrecomp{(\app{\nt^{-1}}{\x})} = (\EPrecomp{(\app{\nt}{\x})})^{-1}$.
\end{proof}

\begin{thrm}[\coqdocurl{Bicategories.DisplayedBicats.Examples.EnrichedCats}{is_univalent_2_bicat_of_enriched_cats}]
\label{thm:univalence-principle-enriched-cats}
If $\V$ is a univalent monoidal category,
then the bicategory $\EnrichCat{\V}$ is univalent.
\end{thrm}

\begin{proof}
It suffices to verify that the displayed bicategory $\dEnrichCat{\V}$ is both locally and globally univalent \cite[Theorem 7.4]{bicatspaper}.
To see that $\dEnrichCat{\V}$ is locally univalent,
we need to show that identities of enrichments $\Ef_1, \Ef_2 : \Ec_1 \onecell \Ec_2$ for a functor $\F$
are the same as invertible 2-cells from $\Ef_1$ to $\Ef_2$ over the identity natural transformation.
We first note that we have $\Ef_1 = \Ef_2$ if and only if we have $\EfunImpl{\Ef_1}{\x}{\y} = \EfunImpl{\Ef_2}{\x}{\y}$ for all $\x, \y : \Ec_1$,
and that invertible 2-cells over the identity are the same as enrichments for the identity by \Cref{prop:inv2cell-enriched}.
Using \cref{eq:nat-trans-enrichment}, we see that the desired enrichment exists if and only if $\EfunImpl{\Ef_1}{\x}{\y} = \EfunImpl{\Ef_2}{\x}{\y}$ for all $\x, \y : \Ec_1$,
which concludes the proof that $\dEnrichCat{\V}$ is locally univalent.

Next we verify that $\dEnrichCat{\V}$ is globally univalent,
meaning that we need to show that identities between enrichments $\Ec_1$ and $\Ec_2$ over a category $\C$
are equivalent to adjoint equivalences between $\Ec_1$ and $\Ec_2$ over the identity.
Since $\V$ is univalent, identities between $\Ec_1$ and $\Ec_2$ correspond to isomorphisms
$i_{\x, \y} : \Ehom{\Ec_1}{\x}{\y} \cong \Ehom{\Ec_2}{\x}{\y}$ for all $\x, \y : \C$
making the following diagrams commute.
\[
\begin{tikzcd}
  & {\munit[V]} \\
  {\Ehom{\Ec_1}{\x}{\x}} && {\Ehom{\Ec_2}{\x}{\x}}
  \arrow["{\EidImpl[\Ec_1]{\x}}"', from=1-2, to=2-1]
  \arrow["{\EidImpl[\Ec_2]{\x}}", from=1-2, to=2-3]
  \arrow["{i_{\x, \x}}"', from=2-1, to=2-3]
\end{tikzcd}
\quad \quad
\begin{tikzcd}
  & {\munit[V]} \\
  {\Ehom{\Ec_1}{\x}{\y}} && {\Ehom{\Ec_2}{\x}{\y}}
  \arrow["{\EFromArr{\f}}"', from=1-2, to=2-1]
  \arrow["{\EFromArr{\f}}", from=1-2, to=2-3]
  \arrow["{i_{\x, \y}}"', from=2-1, to=2-3]
\end{tikzcd}
\]
\[
\begin{tikzcd}[column sep = 4em]
  {\mult{\Ehom{\Ec_1}{\y}{\z}}{\Ehom{\Ec_1}{\x}{\y}}} & {\mult{\Ehom{\Ec_2}{\y}{\z}}{\Ehom{\Ec_2}{\x}{\y}}} \\
  {\Ehom{\Ec_1}{\x}{\z}} & {\Ehom{\Ec_2}{\x}{\z}}
  \arrow["{\mult{i_{\y, \z}}{i_{\x, \y}}}", from=1-1, to=1-2]
  \arrow["{\EcompImpl[\Ec_1]{\x}{\y}{\z}}"', from=1-1, to=2-1]
  \arrow["{\EcompImpl[\Ec_2]{\x}{\y}{\z}}", from=1-2, to=2-2]
  \arrow["{i_{\x, \z}}"', from=2-1, to=2-2]
\end{tikzcd}
\]
In addition,
for every morphism $\f : \munit[V] \onecell \Ehom{\Ec_1}{\x}{\y}$
we have that $\EToArr{\f} = \EToArr{\f \cdot i_{\x , \y}}$.
The inverse $i_{\x, \y}^{-1}$ satisfies similar equalities.
Identities $\Ec_1 = \Ec_2$
give rise to enriched functors $\Ef_1 : \Ec_1 \onecell \Ec_2$ and $\Ef_2 : \Ec_2 \onecell \Ec_1$
such that $\EfunImpl{\Ef_1}{\x}{\y} = i_{\x, \y}$ and $\EfunImpl{\Ef_2}{\x}{\y} = i_{\x, \y}^{-1}$.
This allows us to conclude that $\dEnrichCat{\V}$ is globally univalent.
\end{proof}

\section{Examples of Enriched Categories}
\label{sec:examples}
Before we continue our study of univalent enriched categories,
we first look at numerous examples of enrichments that we use in the remainder of this paper.
In \cref{sec:structure-enrichment}, we characterize enrichments over a large class of structures.

\subsection{General Examples}
\label{sec:general-examples}

\begin{exa}
\label{exa:self-enrichment}
Let $\V$ be a symmetric monoidal closed category.
We define a $\V$-enrichment for $\V$, which we call the \conceptDef{self-enrichment}{CategoryTheory.EnrichedCats.Examples.SelfEnriched}{self_enrichment} and denote by $\self{\V}$, as follows.
\begin{itemize}
  \item We define $\Ehom{\self{\V}}{\x}{\y}$ to be $\mhom{\x}{\y}$.
  \item The enriched identity $\EidImpl{\x} : \munit \onecell \mhom{\x}{\x}$ is defined to be $\mlam{\mlunitImpl{\x}}$.
  \item The composition $\EcompImpl{\x}{\y}{\z} : \mult{(\mhom{\y}{\z})}{(\mhom{\x}{\y})} \onecell \mhom{\x}{\z}$ is the exponential transpose of the following composition of morphisms.
    \[\begin{tikzcd}
	{\mult{(\mult{(\mhom{\y}{\z})}{(\mhom{\x}{\y})})}{\x}} & {\mult{(\mhom{\y}{\z})}{(\mult{(\mhom{\x}{\y})}{\x})}} & {\mult{(\mhom{\y}{\z})}{\y}} & \z
	\arrow["{\massoc{}{}{}}", from=1-1, to=1-2]
	\arrow["{\mult{\id}{\meval{\x}{\y}}}", from=1-2, to=1-3]
	\arrow["{\meval{\y}{\z}}", from=1-3, to=1-4]
      \end{tikzcd}\]
  \item Given $\f : \x \onecell \y$, we define $\EFromArr{\f} : \munit \onecell \mhom{\x}{\y}$ to be $\mlam{\mlunitImpl{\x} \cdot f}$.
  \item For $\f : \munit \onecell \mhom{\x}{\y}$, we define $\EToArr{\f}$ to be the following composition of morphisms.
    \[\begin{tikzcd}
	\x & {\mult{\munit}{\x}} & {\mult{(\mhom{\x}{\y})}{\x}} & \y
	\arrow["{\mlinvunitImpl{x}}", from=1-1, to=1-2]
	\arrow["{\mult{\f}{\id[x]}}", from=1-2, to=1-3]
	\arrow["{\meval{\x}{\y}}", from=1-3, to=1-4]
      \end{tikzcd}\]
\end{itemize}
If we assume that $\V$ is univalent,
then $\self{\V}$ is a univalent enriched category.
\end{exa}

\begin{exa}
\label{exa:full-sub-enrichment}
Let $\C$ be a category together with a $\V$-enrichment $\Ec$,
and let $\Pred$ be a predicate on the objects of $\C$.
From all of this, we obtain a \conceptDef{$\V$-enrichment $\FSubE{\Pred}$ for the full subcategory}{CategoryTheory.EnrichedCats.Examples.FullSubEnriched}{fullsub_enrichment} $\FSub{\Pred}$,
such that $\Ehom{\FSubE{\Pred}}{\x}{\y} \defeq \Ehom{\Ec}{\x}{\y}$.
If $\C$ is univalent,
then the full subcategory of $\C$ is also univalent,
and in that case, this construction gives rise to a univalent enriched category.
\end{exa}

\begin{exa}
\label{exa:opposite-cat-enrichment}
Suppose that $\V$ is a symmetric monoidal category,
and let $\C$ be a category together with a $\V$-enrichment $\Ec$.
We define the $\V$-enrichment $\OpE{\Ec}$, called the \conceptDef{opposite enrichment}{CategoryTheory.EnrichedCats.Examples.OppositeEnriched}{op_enrichment}, for $\Op{\C}$ as follows.
\begin{itemize}
  \item $\Ehom{\OpE{\Ec}}{\x}{\y} \defeq \Ehom{\Ec}{\y}{\x}$;
  \item $\EidImpl[\OpE{\Ec}]{\x} \defeq \EidImpl[\Ec]{\x}$;
  \item $\EcompImpl[\OpE{\Ec}]{\x}{\y}{\z}
    \defeq
    \begin{tikzcd}
	{\mult{\Ehom{\Ec}{\z}{\y}}{\Ehom{\Ec}{\y}{\x}}} & {\mult{\Ehom{\Ec}{\y}{\x}}{\Ehom{\Ec}{\z}{\y}}} & {\Ehom{\Ec}{\z}{\x}}
	\arrow["\msym", from=1-1, to=1-2]
	\arrow["{\Ecomp[\Ec]{\z}{\y}{\x}}", from=1-2, to=1-3]
      \end{tikzcd}$
\end{itemize}
The operations $\EFromArr{\f}$ and $\EToArr{\f}$ in $\OpE{\Ec}$ are inherited from $\Ec$.
In addition, $\OpE{\Ec}$ gives rise to a univalent enriched category if $\C$ is univalent.
\end{exa}

In fact, using \cref{exa:opposite-cat-enrichment} one can construct a duality involution on $\EnrichCat{\V}$.

\begin{exa}
\label{exa:dialgebras-enrichment}
Suppose that $\V$ is a symmetric monoidal category that has equalizers,
and suppose that we have two enriched functors $\Ef_1, \Ef_2 : \Ec_1 \onecell \Ec_2$.
We have the category $\Dialg{\Ef_1}{\Ef_2}$ of dialgebras whose objects are pairs $(\x, \f)$ consisting of an object $\x : \Ec_1$
together with a morphism $\f : \app{\Ef_1}{\x} \onecell \app{\Ef_2}{\x}$.
Morphisms from $(\x, \f)$ to $(\y, \g)$ are morphisms $\h : \x \onecell \y$ such that the following diagram commutes.
\[
  \begin{tikzcd}
    {\app{\Ef_1}{\x}} & {\app{\Ef_1}{\y}} \\
    {\app{\Ef_2}{\x}} & {\app{\Ef_2}{\y}}
    \arrow["\f"', from=1-1, to=2-1]
    \arrow["\g", from=1-2, to=2-2]
    \arrow["{\app{\Ef_1}{\h}}", from=1-1, to=1-2]
    \arrow["{\app{\Ef_2}{\h}}"', from=2-1, to=2-2]
  \end{tikzcd}
\]
We define a \conceptDef{$\V$-enrichment $\DialgE{\Ef_1}{\Ef_2}$ for $\Dialg{\Ef_1}{\Ef_2}$}{CategoryTheory.EnrichedCats.Examples.DialgebraEnriched}{dialgebra_enrichment}.
Suppose that we have objects $(\x, \f)$ and $(\y, \g)$ in $\Dialg{\Ef_1}{\Ef_2}$.
We define the object $\Ehom{\DialgE{\Ef_1}{\Ef_2}}{(\x, \f)}{(\y, \g)}$
as the equalizer of the following diagram.
\[\begin{tikzcd}
	&& {\Ehom{\Ec_2}{\app{\Ef_1}{\x}}{\app{\Ef_1}{\y}}} \\
	{\Ehom{\DialgE{\Ef_1}{\Ef_2}}{(\x, \f)}{(\y, \g)}} & {\Ehom{\Ec_1}{\x}{\y}} && {\Ehom{\Ec_2}{\app{\Ef_1}{\x}}{\app{\Ef_2}{\y}}} \\
	&& {\Ehom{\Ec_2}{\app{\Ef_2}{\x}}{\app{\Ef_2}{\y}}}
	\arrow["{\Efun{\Ef_1}{\x}{\y}}", from=2-2, to=1-3]
	\arrow["{\Efun{\Ef_2}{\x}{\y}}"', from=2-2, to=3-3]
	\arrow["{\EPostcomp{g}}", from=1-3, to=2-4]
	\arrow["{\EPrecomp{f}}"', from=3-3, to=2-4]
	\arrow[hook, from=2-1, to=2-2]
\end{tikzcd}\]
To define the enriched identity and composition morphisms, one uses the universal property of equalizers.
If $\C$ is univalent, then so is the category of dialgebras,
and in that case, $\DialgE{\Ef_1}{\Ef_2}$ is a univalent enriched category.
\end{exa}

From \cref{exa:dialgebras-enrichment}, we get inserters in the bicategory $\EnrichCat{\V}$.

\begin{exa}
\label{exa:enriched-functor-cat}
Let $\V$ be a complete symmetric monoidal category, and suppose that we have enriched categories $\Ec_1$ and $\Ec_2$.
Note that we have a category $\EFunctor{\Ec_1}{\Ec_2}$ whose objects are given by enriched functors from $\Ec_1$ to $\Ec_2$,
and whose morphisms are given by enriched natural transformations.
We define a \conceptDef{$\V$-enrichment $\EFunctorE{\Ec_1}{\Ec_2}$ for $\EFunctor{\Ec_1}{\Ec_2}$}{CategoryTheory.EnrichedCats.Examples.FunctorCategory}{enriched_functor_category_enrichment}
as the equalizer of the morphisms displayed below.
\[
  \begin{tikzcd}
    {\prod_{\x : \Ec_1}\Ehom{\Ec_2}{\app{\Ef_1}{\x}}{\app{\Ef_2}{\x}}} & {\prod_{\x, \y : \Ec_1} \mhom{\Ehom{\Ec_1}{x}{y}}{\Ehom{\Ec_2}{\app{\Ef_1}{\x}}{\app{\Ef_2}{\y}}}}
    \arrow["\g"', shift right=2, from=1-1, to=1-2]
    \arrow["\f", shift left=2, from=1-1, to=1-2]
  \end{tikzcd}
\]
Here $\f$ is defined to be the following composition of morphisms
\[
  \begin{tikzcd}
    {\prod_{\x : \Ec_1}\Ehom{\Ec_2}{\app{\Ef_1}{x}}{\app{\Ef_2}{x}}} & {\Ehom{\Ec_2}{\app{\Ef_1}{\y}}{\app{\Ef_2}{\y}}} & {\mhom{\Ehom{\Ec_1}{x}{y}}{\Ehom{\Ec_2}{\app{\Ef_1}{x}}{\app{\Ef_2}{y}}}}
    \arrow["{\pi_\y}", from=1-1, to=1-2]
    \arrow["\varphi", from=1-2, to=1-3]
  \end{tikzcd}
\]
where $\varphi$ is the exponential transpose of
\[
  \begin{tikzcd}
    {\mult{\Ehom{\Ec_2}{\app{\Ef_1}{\y}}{\app{\Ef_2}{\y}}}{\Ehom{\Ec_1}{\x}{\y}}} &[5pt] {\mult{\Ehom{\Ec_2}{\app{\Ef_1}{\y}}{\app{\Ef_2}{\y}}}{\Ehom{\Ec_2}{\app{\Ef_1}{\x}}{\app{\Ef_1}{\y}}}} &[-2pt] {\Ehom{\Ec_2}{\app{\Ef_1}{\x}}{\app{\Ef_2}{\y}}.}
    \arrow["\mult{\id{}}{\Efun{\Ef_1}{\x}{\y}}", from=1-1, to=1-2]
    \arrow["\Ecomp{}{}{}", from=1-2, to=1-3]
  \end{tikzcd}
\]
We define $g$ analogously.
The fact that $\EFunctor{\Ec_1}{\Ec_2}$ is univalent,
follows from the fact that $\EnrichCat{\V}$ is univalent (\cref{thm:univalence-principle-enriched-cats}).
\end{exa}

Inspired by \cref{exa:enriched-functor-cat}, we can refine \cref{exa:self-enrichment}.
More specifically, given a small category $\C$ and suppose that $\V$ is complete,
we can define an enrichment over the functor category from $\C$ to $\V$.
The construction is analogous to \cref{exa:enriched-functor-cat}, and details can be found in the formalization.

Finally, we look at the change of base operation for enriched categories,
and for this operation, a subtlety arises.
Given a lax monoidal functor $\F : \V_1 \onecell \V_2$,
our goal is to define a pseudofunctor $\EnrichCat{\V_1} \onecell \EnrichCat{\V_2}$.
On objects, this operation acts as follows:
given a univalent category $\C$ together with an enrichment $\Ec$,
then we get an enriched category $\Change{\F}{\Ec}$ whose objects are objects in $\C$
and such that $\Ehom{\Change{\F}{\Ec}}{\x}{\y} = \app{\F}{(\Ehom{\Ec}{\x}{\y})}$.
However, the underlying category of this enriched category is \textbf{not} necessarily univalent.
For instance, if we take $\F$ to be the unique monoidal functor from $\Set$ to the terminal category,
then the underlying category of $\Change{\F}{\Ec}$ would have sets as objects,
and inhabitants of the unit type as the morphisms.
For this reason,
we add a restriction to $\F$ in order to define the change of base of enriched categories.

\begin{defi}
\label{def:preserves-underlying}
Let $\F : \V_1 \onecell \V_2$ be a lax monoidal functor.
We say that $\F$ \conceptDef{preserves underlying categories}{CategoryTheory.EnrichedCats.Examples.ChangeOfBase}{preserve_underlying}
if for all $\x : \V_1$ the function that sends morphisms $\f : \munit[\V_1] \onecell x$
to
$\begin{tikzcd}
  {\munit[\V_2]} & {\app{\F}{\munit[\V_1]}} & {\app{\F}{\x}}
  \arrow["{\mfununit{\F}}", from=1-1, to=1-2]
  \arrow["{\app{\F}{\f}}", from=1-2, to=1-3]
\end{tikzcd}$
is an equivalence of types.
If we have $\f : \munit[\V_2] \onecell \app{\F}{\x}$,
then we denote the action of the inverse by $\mfunchange{\F}{\f}$.
\end{defi}

Every strong monoidal functor that is fully faithful also preserves underlying categories,
because in that case $\mfununit{\F}$ is an isomorphism
and the map sending $\f$ to $\app{\F}{\f}$ is an equivalence of types.
The requirement in \cref{def:preserves-underlying} says that the underlying category is preserved by change of base along $\F$.
With this additional assumption, we define the change of base of enriched categories.

\begin{exa}
\label{exa:change-of-base}
Let $\F : \V_1 \onecell \V_2$ be a lax monoidal functor that preserves underlying categories,
and let $\C$ be a category together with a $\V$-enrichment $\Ec$.
We define the \conceptDef{change-of-base enrichment}{CategoryTheory.EnrichedCats.Examples.ChangeOfBase}{change_of_base_enrichment} $\Change{\F}{\Ec}$ for $\C$ as follows.
The hom-object $\Ehom{\Change{\F}{\Ec}}{\x}{\y}$ is defined to be $\app{\F}{(\Ehom{\Ec}{\x}{\y})}$,
and the enriched identity $\EidImpl{\x}$ is defined as the composition.
\[
  \begin{tikzcd}
    {\munit[\V_2]} &[-2pt] {\app{\F}{\munit[\V_1]}} &[7pt] {\app{\F}{(\Ehom{\Ec}{\x}{\x})}}
    \arrow["{\mfununit{\F}}", from=1-1, to=1-2]
    \arrow["{\app{\F}{(\EidImpl{\x})}}", from=1-2, to=1-3]
  \end{tikzcd}
\]
Composition is defined similarly.
\[
  \begin{tikzcd}
    {\mult{\app{\F}{(\Ehom{\Ec}{\y}{\z})}}{\app{\F}{(\Ehom{\Ec}{\x}{\y})}}} &[-2pt] {\app{\F}{(\mult{\Ehom{\Ec}{\y}{\z}}{\Ehom{\Ec}{\x}{\y}})}} &[7pt] {\app{\F}{(\Ehom{\Ec}{\x}{\y})}}
    \arrow["{\app{\F}{\Ecomp{\x}{\y}{\z}}}", from=1-2, to=1-3]
    \arrow["{\mfunmult{\F}}", from=1-1, to=1-2]
  \end{tikzcd}
\]
If we have a morphism $\f : \x \onecell \y$, then we define $\EFromArr{\f}$ to be
\[
  \begin{tikzcd}
    {\munit[\V_2]} &[-2pt] {\app{\F}{\munit[\V_1]}} &[7pt] {\app{\F}{(\Ehom{\Ec}{\x}{\x})}.}
    \arrow["{\mfununit{\F}}", from=1-1, to=1-2]
    \arrow["{\app{\F}{(\EFromArr{\f})}}", from=1-2, to=1-3]
  \end{tikzcd}
\]
Finally, for a morphism $\f : \munit[\V_2] \onecell \app{\F}{(\Ehom{\Ec}{\x}{\x})}$,
we define $\EToArr{\f}$ to be $\EToArr{\mfunchange{\F}{\f}}$.
Note that here we use that $\F$ preserves underlying categories.
In addition, if we assume that $\C$ is univalent,
then we get a univalent enriched category $\Change{\F}{\Ec}$.
\end{exa}

Finally,
we discuss two classical examples of enriched categories,
namely the unit and the tensor.
Instead of using enrichments, we construct them in the usual way,
from which we still get a category with an enrichment.

\begin{exa}
\label{exa:unit-enriched}
Given a monoidal category $\V$,
we construct the \conceptDef{unit enriched category $\UnitE$}{CategoryTheory.EnrichedCats.Examples.UnitEnriched}{unit_enriched_cat_with_enrichment} as follows.

\begin{itemize}
  \item The objects of $\UnitE$ are inhabitants of the unit type.
  \item The enriched hom object from $\unittel$ to $\unittel$ is $\munit[\V]$ where $\unittel$ is the inhabitant of the unit type.
  \item The enriched identity is given by $\id : \munit[\V] \onecell \munit[\V]$.
  \item The enriched composition is given by $\mlunit{} : \munit[\V] \otimes \munit[\V] \onecell \munit[\V]$.
\end{itemize}
\end{exa}

Note that $\UnitE$ is not necessarily univalent.
For instance,
if $\V$ is the symmetric monoidal category of abelian groups with the tensor product,
then the morphisms in $\UnitE$ from $\unittel$ to $\unittel$ are the same as integers and composition is given by addition.
As a consequence, there is an isomorphism from $\unittel$ to $\unittel$ for each integer,
while for $\UnitE$ to be univalent, there must be a unique isomorphism.

\begin{exa}
\label{exa:tensor-enriched}
Let $\V$ be a symmetric monoidal category, and suppose that we have enriched categories $\Ec_1$ and $\Ec_2$ over $\V$.
We construct their \conceptDef{tensor product $\mult{\Ec_1}{\Ec_2}$}{CategoryTheory.EnrichedCats.Examples.Tensor}{tensor_cat_with_enrichment} as follows.

\begin{itemize}
  \item The objects of $\mult{\Ec_1}{\Ec_2}$
    are pairs $(\x , \y)$ of objects $\x : \Ec_1$ and $\y : \Ec_2$.
  \item The enriched hom object from $(\x_1 , \y_1)$ to $(\x_2 , \y_2)$ is $\mult{\Ehom{\Ec_1}{\x_1}{\x_2}}{\Ehom{\Ec_2}{\y_1}{\y_2}}$.
  \item The enriched identity is given by the following composition of morphisms.
    \[
      \begin{tikzcd}[column sep = 4.5em]
	{\munit[\V]} & {\munit[\V] \otimes \munit[\V]} & {\mult{\Ehom{\Ec_1}{\x}{\x}}{\Ehom{\Ec_2}{\y}{\y}}}
	\arrow["{\mlinvunit{}}", from=1-1, to=1-2]
	\arrow["{\mult{\EidImpl{\x}}{\EidImpl{\y}}}", from=1-2, to=1-3]
      \end{tikzcd}
    \]
  \item The enriched composition is defined as the following composition of morphisms.
    \[
      \begin{tikzcd}
        {(\mult{\Ehom{\Ec_1}{\x_1}{\x_2}}{\Ehom{\Ec_2}{\y_1}{\y_2}}) \otimes (\mult{\Ehom{\Ec_1}{\x_2}{\x_3}}{\Ehom{\Ec_2}{\y_2}{\y_3}})} \\
        {(\mult{\Ehom{\Ec_1}{\x_1}{\x_2}}{\Ehom{\Ec_1}{\x_2}{\x_3}}) \otimes (\mult{\Ehom{\Ec_2}{\y_1}{\y_2}}{\Ehom{\Ec_2}{\y_2}{\y_3}})} \\
        {\mult{\Ehom{\Ec_1}{\x_1}{\x_3}}{\Ehom{\Ec_2}{\y_1}{\y_3}}}
        \arrow["i", from=1-1, to=2-1]
        \arrow["{\mult{\Ecomp{}{}{}}{\Ecomp{}{}{}}}", from=2-1, to=3-1]
      \end{tikzcd}
    \]
    Here the morphism $i$ is defined using that $\V$ is symmetric monoidal.
\end{itemize}
\end{exa}

Let us reflect upon the underlying category of $\mult{\Ec_1}{\Ec_2}$.
Morphisms from $(\x_1 , \y_1)$ to $(\x_2 , \y_2)$ in the underlying category
are given by morphisms $\munit[\V] \onecell \mult{\Ehom{\Ec_1}{\x_1}{\x_2}}{\Ehom{\Ec_2}{\y_1}{\y_2}}$.
However, this is not necessarily the same as pairs of morphisms $\x_1 \onecell \x_2$ and $\y_1 \onecell \y_2$.
For instance, if $\V$ is the symmetric monoidal category of abelian groups with the tensor product,
then morphisms in $\mult{\Ec_1}{\Ec_2}$ are maps $\mathbb{Z} \onecell \mult{\Ehom{\Ec_1}{\x_1}{\x_2}}{\Ehom{\Ec_2}{\y_1}{\y_2}}$.
Such maps are uniquely determined by an element of $\mult{\Ehom{\Ec_1}{\x_1}{\x_2}}{\Ehom{\Ec_2}{\y_1}{\y_2}}$,
which are not the same as pairs of some $\f : \Ehom{\Ec_1}{\x_1}{\x_2}$ and $\g : \Ehom{\Ec_2}{\y_1}{\y_2}$.
For this reason, the tensor product $\mult{\Ec_1}{\Ec_2}$ is not necessarily univalent,
even if we assume $\Ec_1$ and $\Ec_2$ to be so.

\subsection{Enrichments over Structures}
\label{sec:structure-enrichment}
Next we characterize two classes of enrichments.
First, we characterize enrichments for the category $\Set$ of sets equipped with its Cartesian monoidal structure.

\begin{propL}[\coqdocurl{CategoryTheory.EnrichedCats.Examples.SetEnriched}{iscontr_set_enrichment}]
\label{prop:set-enrichment}
Let $\C$ be a category.
The type of $\Set$-enrichments for $\C$ is contractible.
\end{propL}

From \cref{prop:set-enrichment}, we can conclude that the type of categories is equivalent to the type of $\Set$-enriched categories.
Second, we characterize enrichments for structured sets with a Cartesian monoidal structure.
To do so, we first define a general notion of structured sets.

\begin{defi}
\label{def:structure}
A \conceptDef{Cartesian notion of structure}{CategoryTheory.DisplayedCats.Structures.CartesianStructure}{hset_cartesian_struct} $\Struct$ consists of
\begin{itemize}
  \item a set $\app{\StructOb{\Struct}}{\X}$ of structures on $\X$ for every set $\X$;
  \item a proposition $\StructMor{\StructOnOb{\X}}{\StructOnOb{\Y}}{\f}$ which represents that $f$ is a structure preserving map from $\StructOnOb{\X}$ to $\StructOnOb{\Y}$, for all functions $\f : \X \rightarrow \X$ and structures $\StructOnOb{\X} : \app{\StructOb{\Struct}}{\X}$ and $\StructOnOb{\Y} : \app{\StructOb{\Struct}}{\Y}$;
  \item an inhabitant $\StructUnit : \app{\StructOb{\Struct}}{\unitt}$;
  \item a structure $\StructProd{\StructOnOb{\X}}{\StructOnOb{\Y}} : \app{\StructOb{\Struct}}{(\X \times \Y)}$ for all $\StructOnOb{\X} : \app{\StructOb{\Struct}}{\X}$ and $\StructOnOb{\Y} : \app{\StructOb{\Struct}}{\Y}$.
\end{itemize}
This data is required to satisfy the following axioms.
\begin{itemize}
  \item For every set $\X$ and structure $\StructOnOb{\X} : \app{\StructOb{\Struct}}{\X}$, we have $\StructMor{\StructOnOb{\X}}{\StructOnOb{\X}}{\id[\X]}$; 
  \item for all functions $\f : \X \rightarrow \Y$ and $\g : \Y \rightarrow \Z$
    such that $\StructMor{\app{\StructOb{\Struct}}{\X}}{\app{\StructOb{\Struct}}{\Y}}{\f}$
    and $\StructMor{\app{\StructOb{\Struct}}{\Y}}{\app{\StructOb{\Struct}}{\Z}}{\g}$,
    we have $\StructMor{\app{\StructOb{\Struct}}{\X}}{\app{\StructOb{\Struct}}{\Z}}{\g \circ \f}$;
  \item given structures $\StructOnOb{\X}, \StructOnOb{\X}' : \app{\StructOb{\Struct}}{\X}$
    such that $\StructMor{\StructOnOb{\X}}{\StructOnOb{\X'}}{\id[\X]}$ and $\StructMor{\StructOnOb{\X'}}{\StructOnOb{\X}}{\id[\X]}$,
    we have $\StructOnOb{\X} = \StructOnOb{\X}'$;
  \item given a structure $\StructOnOb{\X} : \app{\StructOb{\Struct}}{\X}$ on a set $\X$,
    we have $\StructMor{\StructOnOb{\X}}{\StructUnit}{\lambdatm{(\x : \X)}{\unittel}}$ where $\unittel$ is the unique element of $\unitt$;
  \item given structures $\StructOnOb{\X} : \app{\StructOb{\Struct}}{\X}$ and $\StructOnOb{\Y} : \app{\StructOb{\Struct}}{\Y}$
    on sets $\X$ and $\Y$ respectively,
    we have $\StructMor{\StructProd{\StructOnOb{\X}}{\StructOnOb{\Y}}}{\StructOnOb{\X}}{\pi_1}$
    and $\StructMor{\StructProd{\StructOnOb{\X}}{\StructOnOb{\Y}}}{\StructOnOb{\Y}}{\pi_2}$;
  \item for all functions $\f : \X \rightarrow \Y$ and $\g : \X \rightarrow \Z$
    such that $\StructMor{\app{\StructOb{\Struct}}{\X}}{\app{\StructOb{\Struct}}{\Y}}{\f}$
    and $\StructMor{\app{\StructOb{\Struct}}{\X}}{\app{\StructOb{\Struct}}{\Z}}{\g}$,
    we have $\StructMor{\app{\StructOb{\Struct}}{\X}}{\StructProd{\StructOnOb{\Y}}{\StructOnOb{\Z}}}{\lambdatm{(\x : \X)}{(\app{\f}{\x} , \app{\g}{\x})}}$.
\end{itemize}
A \conceptDef{Cartesian closed notion of structure}{CategoryTheory.DisplayedCats.Structures.StructureLimitsAndColimits}{hset_cartesian_closed_struct} $\Struct$ is given by a Cartesian notion of structure $\Struct$
together with a structure $\StructHom{\StructOnOb{\X}}{\StructOnOb{\Y}} : \app{\StructOb{\Struct}}{(\sigmatype{\f}{\X \rightarrow \Y}{\StructMor{\StructOnOb{\X}}{\StructOnOb{\Y}}{\f}})}$
for all $\StructOnOb{\X} : \app{\StructOb{\Struct}}{\X}$ and $\StructOnOb{\Y} : \app{\StructOb{\Struct}}{\Y}$
such that
\begin{itemize}
  \item the constant map is structure preserving,
    i.e.,
    we have $\StructMor{\StructOnOb{\X}}{\StructOnOb{\Y}}{\lambdatm{\x}{\y}}$ for all $\y : \Y$;
  \item the evaluation is structure preserving,
    i.e.,
    we have $\StructMor{\StructProd{\StructOnOb{\X}}{(\StructHom{\StructOnOb{\X}}{\StructOnOb{\Y}})}}{\StructOnOb{\Y}}{\lambdatm{\z}{\pi_2(\app{\pi_1}{\z}}}$;
  \item lambda abstraction is structure preserving,
    i.e.,
    given $\f : \X \times \Z \rightarrow \Y$ such that $\StructMor{\StructProd{\StructOnOb{\X}}{\StructOnOb{\Z}}}{\StructOnOb{\Y}}{\f}$,
    we also have $\StructMor{\StructOnOb{\Z}}{\StructHom{\StructOnOb{\X}}{\StructOnOb{\Y}}}{\lambda{\z}{\lambdatm{\x}{\f{(\x , \z)}}}}$.
\end{itemize}
Note that for each $\z : \Z$ the map $\lambdatm{\x}{\f{(\x , \z)}}$ is structure preserving,
because the pairing function $\lambdatm{(\x : \X)}{(\app{\f}{\x} , \app{\g}{\x})}$,
the identity,
and constant functions are structure preserving.
\end{defi}

Note that \cref{def:structure} is extension of standard notions of structures defined in \cite[Definition 9.8.1]{hottbook}:
the added data and axioms guarantee that the resulting category has binary products, a terminal object, and exponentials.

\begin{problem}
\label{prob:structure-to-moncat}
Given a Cartesian (closed) notion of structure $\Struct$, to construct a univalent Cartesian (closed) category $\StructCat{\Struct}$.
\end{problem}

\begin{construction}{\coqdocurl{CategoryTheory.Monoidal.Examples.StructuresMonoidal}{monoidal_cat_of_hset_struct}}{prob:structure-to-moncat}
\label{constr:structure-to-moncat}
In \cite[Section 9.8]{hottbook},
it is shown how every standard notion of structure gives rise to a univalent category.
The terminal object is given by $(\unitt, \StructUnit)$,
and the product of $(\X, \StructOnOb{\X})$ and $(\Y, \StructOnOb{\Y})$ is given by $(\X \times \Y, \StructProd{\StructOnOb{\X}}{\StructOnOb{\Y}})$.
Their internal hom is $((\sigmatype{\f}{\X \rightarrow \Y}{\StructMor{\StructOnOb{\X}}{\StructOnOb{\Y}}{\f}}), \StructHom{\StructOnOb{\X}}{\StructOnOb{\Y}})$.
\end{construction}

\begin{propL}[\coqdocurl{CategoryTheory.EnrichedCats.Examples.StructureEnriched}{enrichment_over_struct_weq_struct_enrichment}]
\label{prop:structure-enrichment}
Let $\C$ be a category and let $\Struct$ be a Cartesian notion of structure.
Then the type of $\StructCat{\Struct}$-enrichments for $\C$ is equivalent to
a structure $\HomStruct{\x}{\y} : \app{\StructOb{\Struct}}{(x \onecell y)}$ for all objects $\x, \y : \C$
such that for all $\x, \y, \z : \C$ we have $\StructMor{\StructProd{\HomStruct{\y}{\z}}{\HomStruct{\x}{\y}}}{\HomStruct{\x}{\z}}{\lambdatm{\f}{\app{\pi_2}{\f} \cdot \app{\pi_1}{\f}}}$.
\end{propL}

As such, to give a $\StructCat{\Struct}$-enrichment for $\C$ one needs to endow every hom-set of $\C$
with an $\Struct$-structure
such that the composition operation is a structure preserving map.

\begin{exa}
\label{exa:dcpo-struct}
We have a Cartesian closed notion of structure $\DCPOStruct$ of \conceptDef{directed complete partial orders structures}{CategoryTheory.DisplayedCats.Examples.DCPOStructures}{cartesian_closed_struct_dcpo} (DCPOs)
such that $\app{\StructOb{\DCPOStruct}}{\X}$ is the set of DCPOs on $\X$
and such that $\StructMor{\StructOnOb{\X}}{\StructOnOb{\Y}}{\f}$ expresses that $\f$ is a Scott continuous map.
As such, a $\DCPOStruct$-enriched category is given by a category whose hom-sets are directed complete partial orders,
and whose composition operation is a Scott-continuous map.

We also have a Cartesian notion of structure $\DCPPOStruct$ of \conceptDef{pointed directed complete partial orders structures}{CategoryTheory.DisplayedCats.Examples.PointedDCPOStructures}{cartesian_closed_struct_dcppo}
such that $\app{\StructOb{\DCPPOStruct}}{\X}$ is the set of pointed DCPOs on $\X$
and such that $\StructMor{\StructOnOb{\X}}{\StructOnOb{\Y}}{\f}$ expresses that $\f$ is a Scott continuous map.
Hence, $\DCPPOStruct$-enriched categories are categories whose hom-sets are pointed directed complete partial orders,
and whose composition operation is a Scott-continuous map.
\end{exa}

\subsection{Smash Products}
\label{sec:smash-products}
Finally,
we discuss a notion of structure
that gives rise to a symmetric monoidal closed category
via the smash product.
This notion of structure is tailored to instances where one can construct the smash product using quotients in the category of sets,
such as pointed sets and pointed posets,
and it does not cover examples like pointed DCPOs.
Recall that, given two pointed sets $(\X, *_\X)$ and $(\Y, *_\Y)$,
their smash product $\X \wedge \Y$ is defined to be $(\X \times \Y)/{\sim}$
where the equivalence relation generated by the following clauses.
\[
(*_\X , \y_1) \sim (*_X , \y_2) \quad \quad (\x_1 , *_\Y) \sim (\x_2 , *_\Y)
\]
Before we can define the desired notion of structure,
we first consider a notion of structure that endows every set with a point.

\begin{defi}
\label{def:pointed-structure}
A \conceptDef{pointed notion of structure}{CategoryTheory.DisplayedCats.Structures.StructureLimitsAndColimits}{pointed_hset_struct} is given by
a Cartesian notion $\Struct$ of structure
together with a point $*_{\StructOnOb{\X}} : \X$ for each $\X$ and $\StructOnOb{\X} : \app{\StructOb{\Struct}}{\X}$
such that
\begin{itemize}
  \item the constant map $\lambdatm{x}{*_{\StructOnOb{\Y}}}$ is structure preserving,
    i.e. we have $\StructMor{\StructOnOb{\X}}{\StructOnOb{\Y}}{\lambdatm{x}{*_\Y}}$;
  \item for all functions $\f : \X \rightarrow \Y$
    such that $\StructMor{\StructOnOb{\X}}{\StructOnOb{\Y}}{\f}$,
    we have $\app{\f}{*_{\StructOnOb{\X}}} = *_{\StructOnOb{\Y}}$.
\end{itemize}
\end{defi}

Every pointed notion $\Struct$ of structure gives rise to an equivalence relation $\SmashEqRel$ on $\X \times \Y$
if we have structures $\StructOnOb{\X} : \app{\StructOb{\Struct}}{\X}$ and $\StructOnOb{\Y} : \app{\StructOb{\Struct}}{\Y}$.
More specifically, we say that $(\x_1 , \y_1) \SmashEqRel (\x_2 , \y_2)$ if
either $\x_1 = \x_2$ and $\y_1 = \y_2$
or if we have
\begin{itemize}
  \item either $\x_1 = *_{\StructOnOb{\X}}$ or $\y_1 = *_{\StructOnOb{\Y}}$, and
  \item either $\x_2 = *_{\StructOnOb{\X}}$ or $\y_2 = *_{\StructOnOb{\Y}}$.
\end{itemize}
This gives rise to a set $\SmashStruct{\X}{\Y}{\Struct}$,
which we define to be $(\X \times \Y)/{\SmashEqRel}$.
We write $[( \x , \y )] : \SmashStruct{\X}{\Y}{\Struct}$ for the equivalence class of the pair $(\x , \y)$.

Note that the type $\SmashStruct{\X}{\Y}{\Struct}$ satisfies an elimination principle,
because it is defined as a quotient type.
Given a set $\Z$ and a map $\h : \X \rightarrow \Y \rightarrow \Z$
such that
\begin{equation}
\label{eq:smash-eqs}
\app{\app{\h}{*_{\StructOnOb{\X}}}}{\y_1} = \app{\app{\h}{*_{\StructOnOb{\X}}}}{\y_2}, \quad
\app{\app{\h}{*_{\StructOnOb{\X}}}}{\y} = \app{\app{\h}{\x}}{*_{\StructOnOb{\Y}}}, \quad
\app{\app{\h}{\x_1}}{*_{\StructOnOb{\Y}}} = \app{\app{\h}{\x_2}}{*_{\StructOnOb{\Y}}},
\end{equation}
we have a unique map $\SmashMap{\h} : \SmashStruct{\X}{\Y}{\Struct} \rightarrow \Z$
such that $\app{\SmashMap{\h}}{[( \x , \y )]} = \app{\app{\h}{\x}}{\y}$.

Now we define when a pointed notion of structure supports the smash product.

\begin{defi}
\label{def:smash-structure}
A \conceptDef{notion of structure that supports smash products}{CategoryTheory.DisplayedCats.Structures.StructuresSmashProduct}{hset_struct_with_smash} is given by a pointed notion $\Struct$ of structure together with
\begin{itemize}
  \item a structure $\StructBool{\Struct} : \app{\StructOb{\Struct}}{\bool}$ on the type of booleans;
  \item a structure $\StructSmash{\StructOnOb{\X}}{\StructOnOb{\Y}} : \app{\StructOb{\Struct}}{(\SmashStruct{\X}{\Y}{\Struct})}$
    for all $\StructOnOb{\X} : \app{\StructOb{\Struct}}{\X}$ and $\StructOnOb{\Y} : \app{\StructOb{\Struct}}{\Y}$
\end{itemize}
such that
\begin{itemize}
  \item we have $*_{\StructBool{\Struct}} = \falseB$ and $*_{\StructSmash{\StructOnOb{\X}}{\StructOnOb{\Y}}} = [ (*_{\StructOnOb{\X}} , *_{\StructOnOb{\Y}}) ]$;
  \item we have $\StructMor{\StructBool{\Struct}}{\StructOnOb{\Y}}{\lambdatm{b}{\IfThenElse{b}{\y}{*_{\StructOnOb{\Y}}}}}$
    and $\StructMor{\StructOnOb{\X}}{\StructOnOb{\Y}}{\lambdatm{\x}{\IfThenElse{\app{\f}{\x}}{\app{\g}{\x}}{*_{\StructOnOb{\Y}}}}}$;
  \item we have $\StructMor{\StructOnOb{\X}}{\StructSmash{\StructOnOb{\X}}{\StructOnOb{\Y}}}{\lambdatm{\x}{[ ( \x , \y )]}}$
    and $\StructMor{\StructOnOb{\Y}}{\StructSmash{\StructOnOb{\X}}{\StructOnOb{\Y}}}{\lambdatm{\y}{[ ( \x , \y )]}}$;
  \item we have $\StructMor{\StructProd{\StructOnOb{\X}}{\StructOnOb{\Y}}}{\StructSmash{\StructOnOb{\X}}{\StructOnOb{\Y}}}{\lambdatm{\z}{[ \z ]}}$
    and $\StructMor{\StructSmash{\StructOnOb{\X}}{\StructOnOb{\Y}}}{\StructOnOb{\Z}}{\SmashMap{\h}}$.
\end{itemize}
\end{defi}

To construct the desired monoidal structure,
we also require that our notion of structure is closed.
More specifically,
this means that we require that we have a structure
$\StructHomSmash{\StructOnOb{\X}}{\StructOnOb{\Y}}$
on the hom sets $\StructFunSmash{\StructOnOb{\X}}{\StructOnOb{\Y}}$,
which are defined to be $\sigmatype{\f}{\X \rightarrow \Y}{\StructMor{\StructOnOb{\X}}{\StructOnOb{\Y}}{\f}}$.
We also require that structure preserving maps from $\StructOnOb{\Z}$ to $\StructHomSmash{\StructOnOb{\X}}{\StructOnOb{\Y}}$
are the same as structure preserving from $\StructSmash{\StructOnOb{\Z}}{\StructOnOb{\X}}$ to $\StructOnOb{\Y}$.

Before we precisely define the notion of structure that we use,
we first introduce some notation.
For now we assume that we are given the aforementioned structure $\StructHomSmash{\StructOnOb{\X}}{\StructOnOb{\Y}}$.
The point $*_{\StructHomSmash{\StructOnOb{\X}}{\StructOnOb{\Y}}}$ gives us a structure preserving function from $\StructOnOb{\X}$ to $\StructOnOb{\Y}$.
Given $\f : \Z \rightarrow \StructFunSmash{\StructOnOb{\X}}{\StructOnOb{\Y}}$,
we define the currying map $\app{\SmashCurry}{\f} : \SmashStruct{\Z}{\X}{\Struct} \rightarrow \Y$
to be $\SmashMap{(\lambdatm{\x}{\lambdatm{\y}{\app{\app{\f}{\x}{\y}}}})}$.
Note that $\app{\SmashCurry}{\f}$ only is a function between sets,
and that it is not necessarily structure preserving.
If we have a function $\f : \SmashStruct{\Z}{\X}{\Struct} \rightarrow \Y$
and a point $\z : \Z$,
we define the uncurrying map $\app{\app{\SmashUncurry}{\f}}{\z} : \StructFunSmash{\StructOnOb{\X}}{\StructOnOb{\Y}}$
to be the structure preserving map $\lambdatm{\x}{\app{\f}{[( \z , \x )]}}$.

\begin{defi}
\label{def:smash-structure-closed}
A \conceptDef{notion of closed structure that supports smash products}{CategoryTheory.DisplayedCats.Structures.StructuresSmashProduct}{hset_struct_with_smash_closed}
is given by a pointed notion $\Struct$ of structure together with
a structure $\StructHomSmash{\StructOnOb{\X}}{\StructOnOb{\Y}}$ on the set $\StructFunSmash{\StructOnOb{\X}}{\StructOnOb{\Y}}$
such that the following requirements are satisfied
\begin{itemize}
  \item for all $\x : \X$ we have $\app{*_{\StructHomSmash{\StructOnOb{\X}}{\StructOnOb{\Y}}}}{\x} = *_{\StructOnOb{\Y}}$;
  \item we have
    $\StructMor{\StructSmash{\StructOnOb{\Z}}{\StructOnOb{\X}}}{\StructOnOb{\Y}}{\app{\SmashCurry}{\f}}$
    and $\StructMor{\StructOnOb{\Z}}{\StructHomSmash{\StructOnOb{\X}}{\StructOnOb{\Y}}}{\app{\SmashUncurry}{\f}}$
    for all suitably typed $\f$;
  \item we have
    $\StructMor{\StructHomSmash{\StructOnOb{\Z}}{(\StructHomSmash{\StructOnOb{\X}}{\StructOnOb{\Y}})}}{\StructHomSmash{\StructSmash{\StructOnOb{\Z}}{\StructOnOb{\X}}}{\StructOnOb{\Y}}}{\SmashCurry}$
    and
    $\StructMor{\StructHomSmash{\StructSmash{\StructOnOb{\Z}}{\StructOnOb{\X}}}{\StructOnOb{\Y}}}{\StructHomSmash{\StructOnOb{\Z}}{(\StructHomSmash{\StructOnOb{\X}}{\StructOnOb{\Y}})}}{\SmashUncurry}$.
\end{itemize}
\end{defi}

Let us reflect upon \Cref{def:smash-structure-closed},
and especially the final requirements.
The final two requirements says that we have an isomorphism of \textbf{structures} between
$\StructHomSmash{\StructOnOb{\Z}}{(\StructHomSmash{\StructOnOb{\X}}{\StructOnOb{\Y}})}$
and
$\StructHomSmash{\StructSmash{\StructOnOb{\Z}}{\StructOnOb{\X}}}{\StructOnOb{\Y}}$
rather than a mere isomorphism of sets.
To formulate these two requirements,
we implicitly use that $\app{\SmashCurry}{\f}$ and $\app{\SmashUncurry}{\f}$ are structure preserving.
This is to guarantee that $\SmashCurry$ indeed gives rise to a structure preserving map from
$\StructHomSmash{\StructOnOb{\Z}}{(\StructHomSmash{\StructOnOb{\X}}{\StructOnOb{\Y}})}$
to
$\StructHomSmash{\StructSmash{\StructOnOb{\Z}}{\StructOnOb{\X}}}{\StructOnOb{\Y}}$,
and similarly for $\SmashUncurry$.

\begin{problem}
\label{prob:smash-structure-to-moncat}
Given a notion of closed structure $\Struct$ that supports smash products,
to construct a univalent symmetric monoidal closed category $\StructCatSmash{\Struct}$.
\end{problem}

\begin{construction}{\coqdocurl{CategoryTheory.Monoidal.Examples.SmashProductMonoidal}{smash_product_sym_mon_closed_cat}}{prob:smash-structure-to-moncat}
\label{constr:smash-structure-to-moncat}
The construction is rather technical and is based on a more general construction \cite[Lemma 4.19 and Construction 4.20]{elmendorf2009permutative}.
The main idea is to exploit the fact that we have an isomorphism between
$\StructHomSmash{\StructOnOb{\Z}}{(\StructHomSmash{\StructOnOb{\X}}{\StructOnOb{\Y}})}$
and
$\StructHomSmash{\StructSmash{\StructOnOb{\Z}}{\StructOnOb{\X}}}{\StructOnOb{\Y}}$
in the category $\StructCat{\Struct}$.
\end{construction}

Following \Cref{prop:structure-enrichment}, we can analogously characterize enrichments over the symmetric monoidal closed category $\StructCatSmash{\Struct}$.

\begin{propL}[\coqdocurl{CategoryTheory.EnrichedCats.Examples.SmashStructureEnriched}{enrichment_over_smash_struct_weq_smash_struct_enrichment}]
\label{prop:smash-structure-enrichment}
Let $\C$ be a category and let $\Struct$ be a notion of closed structure that supports smash products.
Then $\StructCat{\Struct}$-enrichments for $\C$ is equivalent to
a structure $\HomStruct{\x}{\y} : \app{\StructOb{\Struct}}{(x \onecell y)}$ for all objects $\x, \y : \C$
such that for all $\x, \y, \z : \C$ we have
\[
\StructMor{\StructProd{\HomStruct{\y}{\z}}{\HomStruct{\x}{\y}}}{\HomStruct{\x}{\z}}{\lambdatm{\f}{\app{\pi_2}{\f} \cdot \app{\pi_1}{\f}}}
\]
and such that for all $\f : \x \rightarrow \y$ we have
\[
\f \cdot *_{\HomStruct{\y}{\z}} = *_{\HomStruct{\x}{\z}}, \quad \quad
*_{\HomStruct{\w}{\x}} \cdot \f = *_{\HomStruct{\w}{\y}}.
\]
\end{propL}

Note that in \Cref{prop:smash-structure-enrichment},
we require composition to be a structure preserving map whose domain is the product
rather than the smash product.
This is fine,
because we can lift composition to an operation on the smash product
due to the final two requirements.
Let us finish this section with two examples of notions of closed structures that support smash products.

\begin{exa}
\label{exa:pointed-set}
\label{exa:pointed-poset}
We define the \conceptDef{notion $\PointedSet$ of structure given by pointed sets}{CategoryTheory.DisplayedCats.Examples.PointedSetStructures}{pointed_struct_pointed_hset_with_smash_closed}.
The type $\app{\StructOb{\PointedSet}}{\X}$ is the collection of points in $\X$,
and $\StructMor{\x}{\y}{\f}$ says that $\f : \X \rightarrow \Y$ preserves the point, i.e. $\app{\f}{\x} = \y$.
One can show that $\PointedSet$ is a notion of closed structure that supports smash products.
Analogously, we define \conceptDef{$\PointedPoset$}{CategoryTheory.DisplayedCats.Examples.PointedPosetStrict}{pointed_struct_pointed_poset_strict_with_smash_closed},
which is given by pointed posets (posets with a least element).
The type $\app{\StructOb{\PointedPoset}}{\X}$ is the set of partial orders for $\X$ that have a least element,
and $\StructMor{\x}{\y}{\f}$ says that $\f : \X \rightarrow \Y$ is strictly monotone.
This also gives rise to a notion of closed structure that supports smash products.
\end{exa}

\begin{rem}
\label{rem:dcpo-smash-prod}
In \cref{exa:dcpo-struct},
we defined a Cartesian notion of structure by \emph{pointed} DCPOs and Scott continuous maps without requiring these maps to be strict.
For pointed DCPOs and strict Scott continuous maps, one can also define such a structure.
However, in applications, one is often interested in a different monoidal structure for pointed DCPOs with strict maps,
namely the one given by the smash product.
While the smash products of pointed DCPOs can be constructed via quotients~\cite[Theorem 2.9.1]{townsend1996preframe},
one must take these quotients in the appropriate category~\cite{sterling:2024}.
Our methods do not extend to pointed DCPOs,
because quotients of sets do not give rise to quotients of (pointed) DCPOs.
\end{rem}

\section{Image Factorization}
\label{sec:image-factorization}
We continue our study of univalent enriched categories by proving that every essentially surjective and fully faithful (enriched) functor is an adjoint equivalence.
Classically, one would use the axiom of choice to prove this fact:
to define the inverse, one needs to pick preimages and those are only guaranteed to be unique up to isomorphism.
One can give a constructive proof of this fact if one assumes that the domain of the functor in question is univalent.

The way we approach this result, is via \emph{orthogonal factorization systems} in bicategories.
More specifically,
we show that the essentially surjective and the fully faithful enriched functors form an orthogonal factorization system \cite[Lemma 4.3.5]{MR4177953}.
From this fact, one directly obtains that every essentially surjective and fully faithful functor is an adjoint equivalence.
The proof is similar to how in orthogonal factorization systems in categories the intersection of the left and right class of maps are precisely the isomorphisms.

We start by defining \emph{orthogonal} maps in bicategories.

\begin{defi}
\label{def:orthogonal-maps}
Let $\B$ be a bicategory and let $\f : \x_1 \onecell \x_2$ and $\g : \y_1 \onecell \y_2$ be 1-cells.
Then we say that $\f$ is \conceptDef{orthogonal}{Bicategories.OrthogonalFactorization.Orthogonality}{orthogonal} to $\g$, written $\orthogonal{\f}{\g}$,
if the following diagram of categories is a weak pullback in the bicategory of categories.
\[
  \begin{tikzcd}
    {\homC{\B}{\x_2}{\y_1}} & {\homC{\B}{\x_1}{\y_1}} \\
    {\homC{\B}{\x_2}{\y_2}} & {\homC{\B}{\x_1}{\y_2}}
    \arrow["{\precomp{\f}}", from=1-1, to=1-2]
    \arrow["{\postcomp{\g}}", from=1-2, to=2-2]
    \arrow["{\postcomp{\g}}"', from=1-1, to=2-1]
    \arrow["{\precomp{\f}}"', from=2-1, to=2-2]
  \end{tikzcd}
\]
where the functors $\precomp{\f}$ and $\postcomp{\g}$ are given by precomposition with $f$ and postcomposition with $g$ respectively.
\end{defi}

Let us reflect on \cref{def:orthogonal-maps}.
Weak pullbacks of categories are given by iso-comma categories.
The objects in the iso-comma category $\IsoComma{\F}{\G}$ of functors $\F : \C_1 \onecell \C_3$ and $\G : \C_2 \onecell \C_3$ are given by triples $(\x, \y, \f)$ of objects $\x : \C_1$ and $\y : \C_2$ together with an isomorphism $\f : \iso{\app{\F}{\x}}{\app{\G}{\y}}$.
Note that we have a functor $\OFunc{\f}{\g} : \homC{\B}{\x_2}{\y_1} \onecell \IsoComma{\precomp{\f}}{\postcomp{\g}}$.
The functor $\OFunc{\f}{\g}$ maps 1-cells $\h : \x_2 \onecell \y_1$ to the triple $(\h \cdot \g, \f \cdot \h, \lassociatorfull{\f}{\h}{\g})$ where $\lassociator{}{}{}$ is the associator of $\B$.
Orthogonality can equivalently be phrased by saying that the functor $\OFunc{\f}{\g}$ is an adjoint equivalence.
Essential surjectivity of $\OFunc{\f}{\g}$ says that every square has a diagonal filler as follows.
\[
  \begin{tikzcd}
    {\x_1} & {\y_1} \\
    {\x_2} & {\y_2}
    \arrow["\g", from=1-2, to=2-2]
    \arrow["\f"', from=1-1, to=2-1]
    \arrow["{\h_1}", from=1-1, to=1-2]
    \arrow["{\h_2}"', from=2-1, to=2-2]
    \arrow["l"{description}, dashed, from=2-1, to=1-2]
  \end{tikzcd}
\]
More concretely, given the diagram above, there is a lift $l : \x_2 \onecell \y_1$ making the two triangles commute up to invertible 2-cell.
Fully faithfulness of $\OFunc{\f}{\g}$ says that whenever we have two lifts $l_1, l_2 : \x_2 \onecell \y_1$ together with 2-cells $\tau_1 : l_1 \cdot \g \twocell l_2 \cdot \g$ and $\tau_2 : f \cdot l_1 \twocell f \cdot l_2$,
we have a unique 2-cell $\zeta : l_1 \twocell l_2$
such that $\zeta \whiskerr g = \tau_1$ and $f \whiskerl \zeta = \tau_2$.

\begin{defi}
\label{def:orthogonal-factorization}
Let $\B$ be a bicategory.
An \conceptDef{orthogonal factorization system}{Bicategories.OrthogonalFactorization.FactorizationSystem}{orthogonal_factorization_system} on $\B$ consists of two classes of maps, which we denote by $\Left$ and $\Right$, such that
\begin{itemize}
  \item $\Left$ and $\Right$ are closed under invertible 2-cells;
  \item for all 1-cells $\f$ and $\g$ such that $\app{\Left}{\f}$ and $\app{\Right}{\g}$, we have $\orthogonal{\f}{\g}$;
  \item for every 1-cell $\f$, we have a factorization $\iso{\f}{l \cdot r}$ such that $\app{\Left}{l}$ and $\app{\Right}{r}$.
\end{itemize}
\end{defi}

Note that if $\B$ is locally univalent,
then all classes of 1-cells are closed under invertible 2-cells.
In this section, we are interested in a particular factorization system on $\EnrichCat{\V}$,
which is given by the fully faithful and the essentially surjective enriched functors.

\begin{defi}
\label{defi:ff-enriched}
Let $\Ef : \Ec_1 \onecell \Ec_2$ be an enriched functor.
\begin{itemize}
  \item We say that $\Ef$ is \conceptDef{fully faithful}{CategoryTheory.EnrichedCats.EnrichmentFunctor}{fully_faithful_enriched_functor} if for all objects $\x, \y : \Ec_1$ the morphism $\EfunImpl{\Ef}{\x}{\y}$ is an isomorphism.
  \item We say that $\Ef$ is \conceptDef{essentially surjective}{CategoryTheory.Core.Functors}{essentially_surjective} if its underlying functor is essentially surjective.
    That is to say, for all $\y : \Ec_2$ we have an inhabitant of $\trunc{\sigmatype{\x}{\Ec_1}{\iso{\app{\Ef}{\x}}{\y}}}$.
  \item We say that $\Ef$ is a \conceptDef{weak equivalence}{CategoryTheory.EnrichedCats.EnrichmentFunctor}{enriched_weak_equivalence} if $\Ef$ is both fully faithful and essentially surjective.
\end{itemize}
\end{defi}

Every enriched functor can be factorized as an essentially surjective functor followed by a fully faithful functor by taking the \emph{full image}.

\begin{exa}
\label{exa:image-enriched}
Let $\Ef : \Ec_1 \onecell \Ec_2$ be an enriched functor.
We define a predicate $\Pred$ on the objects of $\Ec_2$ such that
$
\app{\Pred}{\y} \defeq \trunc{\sigmatype{\x}{\Ec_1}{\iso{\app{\Ef}{\x}}{\y}}}. 
$
The \conceptDef{full image}{CategoryTheory.EnrichedCats.Examples.ImageEnriched}{image_enrichment} $\ImE{\Ef}$ of $\Ef$ is defined to be the full subcategory of $\Ec_2$ with respect to $\Pred$.
\end{exa}

\begin{propL}[\coqdocurl{Bicategories.OrthogonalFactorization.EnrichedEsoFactorization}{enriched_eso_ff_orthogonal}]
\label{prop:eso-ff-orthogonal}
Suppose that we have univalent enriched categories $\Ec_1, \Ec_2, \Ec_3$, and $\Ec_4$.
If we have enriched functors $\Ef : \Ec_1 \onecell \Ec_2$ and $\Eg : \Ec_3 \onecell \Ec_4$ such that $\Ef$ is essentially surjective and $\Eg$ is fully faithful,
then $\orthogonal{\Ef}{\Eg}$.
\end{propL}

\begin{problem}
\label{prob:enriched-factorization}
To construct an orthogonal factorization system on $\EnrichCat{\V}$.
\end{problem}

\begin{construction}{\coqdocurl{Bicategories.OrthogonalFactorization.EnrichedEsoFactorization}{enriched_eso_ff_orthogonal_factorization_system}}{prob:enriched-factorization}
\label{constr:image-factorization}
The classes $\Left$ and $\Right$ are given by the essentially surjective and the fully faithful enriched functors respectively.
The desired factorization is given by \cref{exa:image-enriched},
and the proof of orthogonality is given in \cref{prop:eso-ff-orthogonal}.
It remains to show that essentially surjective and fully faithful enriched functors are closed under enriched natural isomorphisms,
and this is so, because $\EnrichCat{\V}$ is locally univalent.
\end{construction}

From this factorization system, we directly obtain that weak equivalence are actually adjoint equivalences.

\begin{thrm}[\coqdocurl{Bicategories.OrthogonalFactorization.EnrichedEsoFactorization}{enriched_eso_ff_adjoint_equivalence}]
\label{thm:enriched-factorization}
Every fully faithful and essentially surjective enriched functor $\Ef : \Ec_1 \onecell \Ec_2$ is an adjoint equivalence.
\end{thrm}

\begin{proof}
Suppose that $\Ef : \Ec_1 \onecell \Ec_2$ is fully faithful and essentially surjective.
Consider the following diagram.
\[
  \begin{tikzcd}
    {\Ec_1} & {\Ec_1} \\
    {\Ec_2} & {\Ec_2}
    \arrow["\Ef"', from=1-1, to=2-1]
    \arrow["\Ef", from=1-2, to=2-2]
    \arrow["{\id{}}", from=1-1, to=1-2]
    \arrow["{\id{}}"', from=2-1, to=2-2]
    \arrow["l"{description}, dashed, from=2-1, to=1-2]
  \end{tikzcd}
\]
Due to the orthogonality of fully faithful and essentially surjective morphisms,
this diagram has a lift $l$ such that both triangles commute up to invertible 2-cell.
From this, we get that $\Ef$ is an equivalence,
and since equivalences can be refined to adjoint equivalences,
$\Ef$ is an adjoint equivalence.
\end{proof}

\begin{rem}
\label{rem:domain-univalence}
For the proof of \cref{thm:enriched-factorization},
we require both the domain and codomain of $\Ef$ to be univalent.
This restriction is a consequence of the bicategorical machinery,
because we phrase everything in the bicategory $\EnrichCat{\V}$ whose objects are univalent enriched categories.
In the case that only the domain of $\Ef$ is univalent, one could still use the same construction as in \cref{constr:image-factorization}.
\end{rem}

\section{The Rezk Completion}
\label{sec:enriched-rezk}
The next aspect in our study of univalent enriched categories, is the \emph{enriched Rezk completion}.
There are two features to a suitable Rezk completion for enriched categories.
The first feature of the Rezk completion that we consider, is that every enriched category is weakly equivalent to a univalent one.

\begin{problem}
\label{prob:enriched-rezk-completion}
Given a category $\Ec$ enriched over a univalent complete symmetric monoidal category $\V$,
to construct a univalent enriched category $\Rezk{\Ec}$ and a weak equivalence $\RezkFun : \Ec \onecell \Rezk{\Ec}$.
\end{problem}

We give two constructions for \Cref{prob:enriched-rezk-completion}.
The first construction (\Cref{constr:enriched-rezk-completion}) is similar to the Rezk completion of categories \cite{rezk_completion}.
More specifically, we can construct the enriched Rezk completion as the image of the Yoneda embedding.
This construction crucially relies on the assumption that $\V$ is both complete and monoidal closed,
which we need to construct an enrichment for the enriched functor category $\EFunctor{\OpE{\Ec}}{\self{\V}}$.
In addition, this construction increases the universe level in general.

The second construction (\Cref{constr:enriched-rezk-completion-HIT}) uses higher inductive types (HITs),
and we rely on the fact that the Rezk completion of ordinary categories can be desired as a HIT \cite{hottbook}.
Explicitly, if we have a category $\C$ enriched over a univalent monoidal category $\V$,
then we endow its Rezk completion (of categories) $\Rezk{\C}$ with a $\V$-enrichment.
To do so, we use the elimination rules of higher inductive types.
This allows us to avoid the assumption that $\V$ is both complete and symmetric monoidal closed,
and, if we also assume that the universe is closed under HITs,
the universe level remains the same.

We assume that $\V$ is univalent in both \Cref{constr:enriched-rezk-completion,constr:enriched-rezk-completion-HIT},
and we can extend both constructions to the case where $\V$ is not necessarily univalent.
To do so, we first take the Rezk completion of monoidal categories \cite{DBLP:conf/types/WullaertMA22} of $\V$ to obtain a weak equivalence $\RezkFun : \V \onecell \Rezk{\V}$.
Since $\RezkFun$ is fully faithful, it preserves underlying categories.
Hence, if we have a category $\Ec$ enriched over $\V$,
we obtain a category $\Change{\RezkFun}{\Ec}$ enriched over $\Rezk{\V}$ using \cref{exa:change-of-base}.
Then we can use either \Cref{constr:enriched-rezk-completion} or \Cref{constr:enriched-rezk-completion-HIT} to obtain the desired Rezk completion.

The second feature of the Rezk completion that we consider, is the universal property (\cref{thm:rezk-completion-ump}).
We formulate this universal property in such a way that we can instantiate it directly for both \Cref{constr:enriched-rezk-completion,constr:enriched-rezk-completion-HIT}.
More specifically, we say that every enriched functor from an enriched category $\Ec$ to some univalent enriched category
can be extended along weak equivalences.

\begin{thrm}[\coqdocurl{CategoryTheory.EnrichedCats.RezkCompletion.RezkUniversalProperty}{enriched_rezk_completion_ump}]
\label{thm:rezk-completion-ump}
Suppose that we have an enriched functor $\Ef : \Ec_1 \onecell \Ec_2$ that is both essentially surjective and fully faithful.
Then the precomposition functor $\precomp{\Ef} : \EFunctor{\Ec_2}{\Ec_3} \onecell \EFunctor{\Ec_1}{\Ec_3}$ is an adjoint equivalence of categories
for each univalent enriched category $\Ec_3$.
\end{thrm}

To understand \Cref{thm:rezk-completion-ump},
let us first note that $\precomp{\Ef} : \EFunctor{\Ec_2}{\Ec_3} \onecell \EFunctor{\Ec_1}{\Ec_3}$ is both essentially surjective and fully faithful
if it is an adjoint equivalence.
Essential surjectivity expresses that for each enriched functor $\Eg : \Ec_1 \onecell \Ec_3$,
we can find an enriched functor $\Eh : \Ec_2 \onecell \Ec_3$ and a natural isomorphism $\nt$ as follows.
\[
\begin{tikzcd}
  {\Ec_1} & {\Ec_3} \\
  {\Ec_2}
  \arrow[""{name=0, anchor=center, inner sep=0}, "\Eg", from=1-1, to=1-2]
  \arrow["\Ef"', from=1-1, to=2-1]
  \arrow[""{name=1, anchor=center, inner sep=0}, "\Eh"', dashed, from=2-1, to=1-2]
  \arrow["\nt"', shorten <=2pt, shorten >=2pt, Rightarrow, from=0, to=1]
\end{tikzcd}
\]
To understand fully faithfulness we assume that we have diagrams as follows
\[
\begin{tikzcd}
  {\Ec_1} & {\Ec_3} \\
  {\Ec_2}
  \arrow[""{name=0, anchor=center, inner sep=0}, "{\Eg_1}", from=1-1, to=1-2]
  \arrow["\Ef"', from=1-1, to=2-1]
  \arrow[""{name=1, anchor=center, inner sep=0}, "{\Eh_1}"', dashed, from=2-1, to=1-2]
  \arrow["\nt_1"', shorten <=2pt, shorten >=2pt, Rightarrow, from=0, to=1]
\end{tikzcd}
\quad \quad
\begin{tikzcd}
  {\Ec_1} & {\Ec_3} \\
  {\Ec_2}
  \arrow[""{name=0, anchor=center, inner sep=0}, "{\Eg_2}", from=1-1, to=1-2]
  \arrow["\Ef"', from=1-1, to=2-1]
  \arrow[""{name=1, anchor=center, inner sep=0}, "{\Eh_2}"', dashed, from=2-1, to=1-2]
  \arrow["\nt_2"', shorten <=2pt, shorten >=2pt, Rightarrow, from=0, to=1]
\end{tikzcd}
\]
where $\nt_1$ and $\nt_2$ are natural isomorphisms.
In that situation, each enriched natural transformation $\theta : \Eg_1 \twocell \Eg_2$
gives rise to a unique enriched natural transformation $\zeta : \Eh_1 \twocell \Eh_2$
satisfying the following equality.
\[
\begin{tikzcd}
  {\Ec_1} & {\Ec_2} & {\Ec_3}
  \arrow["\Ef"', from=1-1, to=1-2]
  \arrow[""{name=0, anchor=center, inner sep=0}, "{\Eg_1}", bend left=60, from=1-1, to=1-3]
  \arrow[""{name=1, anchor=center, inner sep=0}, "{\Eh_1}", shift left=3, from=1-2, to=1-3]
  \arrow[""{name=2, anchor=center, inner sep=0}, "{\Eh_2}"', shift right=3, from=1-2, to=1-3]
  \arrow["{\tau_1}", shorten <=3pt, Rightarrow, from=0, to=1-2]
  \arrow["\zeta", shorten <=2pt, shorten >=2pt, Rightarrow, from=1, to=2]
\end{tikzcd}
\quad = \quad
\begin{tikzcd}
  {\Ec_1} & {\Ec_2} & {\Ec_3}
  \arrow["\Ef"', from=1-1, to=1-2]
  \arrow[""{name=0, anchor=center, inner sep=0}, "{\Eg_1}", bend left=90, from=1-1, to=1-3]
  \arrow[""{name=1, anchor=center, inner sep=0}, "{\Eg_2}"{description}, bend left=40, from=1-1, to=1-3]
  \arrow["{\Eh_2}"', from=1-2, to=1-3]
  \arrow["{\nt_2}", shorten <=4pt, Rightarrow, from=1, to=1-2]
  \arrow["\theta", shorten <=2pt, shorten >=3pt, Rightarrow, from=0, to=1]
\end{tikzcd}
\]

\subsection{The Rezk Completion via Yoneda}
Now we give a first construction for \Cref{prob:enriched-rezk-completion},
and to do so, we use the Yoneda lemma.
Let us start by first defining enriched representable presheaves and the Yoneda embedding.

\begin{defi}
\label{def:yoneda-embedding}
Let $\V$ be a complete symmetric monoidal closed category,
and let $\Ec$ be a $\V$-enriched category.
Given an object $\y : \Ec$, we define the \conceptDef{representable functor}{CategoryTheory.EnrichedCats.Examples.Yoneda}{enriched_repr_presheaf} $\RepFun{\y} : \EFunctor{\OpE{\Ec}}{\self{\V}}$ as follows.
\begin{itemize}
  \item For objects $\x : \Ec$, we define $\app{\RepFun{\y}}{\x} \defeq \Ehom{\Ec}{\x}{\y}$;
  \item for morphisms $\f : \x_1 \onecell \x_2$, we define $\app{\RepFun{\y}}{\f}$ to be $\EPrecomp{\f} : \Ehom{\Ec}{\x_2}{\y} \onecell \Ehom{\Ec}{\x_1}{\y}$.
\end{itemize}
Given a morphism $\f : \y_1 \onecell \y_2$ in $\Ec$, we define the \conceptDef{representable natural transformation}{CategoryTheory.EnrichedCats.Examples.Yoneda}{enriched_repr_nat_trans} $\RepNat{\f} : \RepFun{\y_1} \onecell \RepFun{\y_2}$ to be $\EPostcomp{\f} : \Ehom{\Ec}{\x}{\y_1} \onecell \Ehom{\Ec}{\x}{\y_2}$ for every $x : \Ec$.

Finally, the \conceptDef{enriched Yoneda embedding}{CategoryTheory.EnrichedCats.Examples.Yoneda}{enriched_yoneda} $\Yon{\Ec} : \Ec \onecell \EFunctor{\OpE{\Ec}}{\self{\V}}$ is defined to be $\RepFun{\y}$ on objects $y : \Ec$ and $\RepNat{\f}$ on morphisms $\f : \y_1 \onecell \y_2$.
\end{defi}

Note that in \cref{def:yoneda-embedding},
one also needs to construct $\V$-enrichments for $\RepFun{\y}$ and $\Yon{\Ec}$,
and prove that $\RepNat{\f}$ is $\V$-enriched.
The details can be found in the literature \cite{kelly1982basic} and in the formalization.

\begin{propL}[\coqdocurl{CategoryTheory.EnrichedCats.YonedaLemma}{fully_faithful_enriched_yoneda}]
\label{prop:yoneda-lemma}
The enriched Yoneda embedding is fully faithful.
\end{propL}

\begin{construction}{\coqdocurl{CategoryTheory.EnrichedCats.RezkCompletion.EnrichedRezkCompletion}{enriched_rezk_completion_bundled}}{prob:enriched-rezk-completion}
\label{constr:enriched-rezk-completion}
We define $\Rezk{\Ec}$ to be the image of the Yoneda embedding $\Yon{\Ec}$.
The enriched functor $\RezkFun : \Ec \onecell \Rezk{\Ec}$ is essentially surjective by construction (\cref{exa:image-enriched}).
By the Yoneda lemma (\cref{prop:yoneda-lemma}), $\RezkFun$ is fully faithful as well.
\end{construction}

Note that \cref{constr:enriched-rezk-completion} might increase the universe level.
Let us assume that the type of objects of $\Ec$ and $\V$ live in $\UnivU$ and $\UnivV$ respectively.
Since objects of $\Rezk{\Ec}$ are enriched presheaves from $\Ec$ to $\self{\V}$ that are in the image of $\Yon{\Ec}$,
the type of objects of $\Rezk{\Ec}$ lives in $\maxU{\UnivU}{\UnivV}$.
In many examples, $\UnivV$ is a larger universe that $\UnivU$,
because we require the category $\V$ to have products indexed by the type of the objects in $\Ec$.

\subsection{The Rezk Completion via HITs}
Now we present another construction for \Cref{prob:enriched-rezk-completion},
but this time we use higher inductive types.
Recall that higher inductive types are specified by
a list of point constructors, path constructors, homotopy constructors, and so on.
Given a category $\C$, we are interested in the 1-truncated HIT $\RezkH{\C}$ with
a single point constructor $\rcl{\C} : \C \rightarrow \RezkH{\C}$,
a path constructor $\app{\rcleq{\C}}{\f} : \app{\rcl{\C}}{\x} = \app{\rcl{\C}}{\y}$ for each isomorphism $f : \iso{\x}{\y}$,
and homotopy constructors
$\app{\re}{\x} : \rcleq{\id{\x}} = \idpath$ for each $\x : \C$
and $\app{\app{\rconcat}{\f}}{\g} : \app{\rcleq{\C}}{(\f \cdot \g)} = \app{\rcleq{\C}}{\f} \cdot \app{\rcleq{\C}}{\g}$
for all isomorphisms $f : \iso{\x}{\y}$ and $g : \iso{\y}{\z}$.
We specify this type precisely using introduction and elimination rules.
Note that the desired HIT can be constructed using mild assumptions \cite{rijke:2018,VvdW}.

Before we give the exact rules of the HIT,
we introduce some notation.
Given a path $p : \x = \y$ in some type $\X$
and a type $\Y$ depending on $\X$,
we write $\xx =_{p} \yy$ for the type of paths from $\xx : \app{\Y}{\x}$ to $\yy : \app{\Y}{\y}$ over $p$.
We have a path $\idpathD[\xx] : \xx =_{\idpath[\x]} \xx$
and for all paths $\pp : \xx =_{p} \yy$ and $\qq : \yy =_{q} \zz$,
we have a concatenation $\pp \cdot \qq : \xx =_{p \cdot q} \zz$.
Finally,
given a homotopy $h : p = q$,
we have a type $\pp =_{h} \qq$ of homotopies over $h$
where $\pp : \xx =_{p} \yy$ and $\qq : \xx =_{q} \yy$.

\begin{figure}
\begin{center}
\begin{bprooftree}
\AxiomC{$\x : \C$}
\UnaryInfC{$\app{\rcl}{\x} : \RezkH{\C}$}
\end{bprooftree}
\begin{bprooftree}
\AxiomC{$\x, \y : \C$}
\AxiomC{$\f : \iso{\x}{\y}$}
\BinaryInfC{$\app{\rcleq}{\f} : \app{\rcl}{\x} = \app{\rcl}{\y}$}
\end{bprooftree}
\begin{bprooftree}
\AxiomC{}
\UnaryInfC{$\Rezk{\C}$ is a 1-type}
\end{bprooftree}
\end{center}
\begin{center}
\begin{bprooftree}
\AxiomC{$\x : \C$}
\UnaryInfC{$\app{\re}{\x} : \rcleq{\id{\x}} = \idpath$}
\end{bprooftree}
\begin{bprooftree}
\AxiomC{$\x, \y, \z : \C$}
\AxiomC{$\f : \iso{\x}{\y}$}
\AxiomC{$\g : \iso{\y}{\z}$}
\TrinaryInfC{$\app{\app{\rconcat}{\f}}{\g} : \app{\rcleq{\C}}{(\f \cdot \g)} = \app{\rcleq{\C}}{\f} \cdot \app{\rcleq{\C}}{\g}$}
\end{bprooftree}
\end{center}
\caption{Introduction rules for $\RezkH{\C}$}
\label{fig:rezk-hit-intro-rules}
\end{figure}

\begin{axiom}
\label{ax:rezk-completion-hit}
We assume that for each category $\C$ we have a type $\RezkH{\C}$ satisfying the introduction rules in \Cref{fig:rezk-hit-intro-rules}.
In addition, we assume if
\begin{itemize}
  \item we have a 1-type $\Y$ depending on $\RezkH{\C}$;
  \item an element $\app{r}{\x} : \app{\Y}{(\app{\rcl}{\x})}$ for each $\x : \C$;
  \item a path $\app{e}{\f} : \app{r}{\x} =_{\app{\rcleq}{\f}} \app{r}{\y}$ for each isomorphism $\f : \iso{\x}{\y}$;
  \item a homotopy $h_1 : e(\id{x}) =_{\app{\re}{\x}} \idpathD[\app{r}{\x}]$ for each $\x : \C$;
  \item a homotopy $h_2 : e(\f \cdot \g) =_{\app{\app{\rconcat}{\f}}{\g}} \app{e}{\f} \cdot \app{e}{\g}$ for all isomorphisms $\f : \iso{\x}{\y}$ and $\g : \iso{\y}{\z}$,
\end{itemize}
then we have a dependent function $\rInd{\Y}{r}{e}{h_1}{h_2} : \pitype{\x}{\RezkH{\C}}{\app{\Y}{\x}}$
for which the following equalities hold.
\[
  \app{\rInd{\Y}{r}{e}{h_1}{h_2}}{(\app{\rcl}{\x})} \equiv \app{r}{\x}
\]
\[
  \apd{\rInd{\Y}{r}{e}{h_1}{h_2}}{(\app{\rcleq}{\f})} = \app{e}{\f}
\]
\end{axiom}

\begin{construction}{\coqdocurlR{sem.rezk.enriched}{enriched_rezk_completion_bundled}}{prob:enriched-rezk-completion}
\label{constr:enriched-rezk-completion-HIT}
Let $\Ec$ be an $\V$-enrichment over $\C$.
We only show how to construct the necessary hom objects and the enriched identity,
because the other operations are defined in a similar way.
First,
we construct the hom objects $\Ehom{\RezkH{\C}}{\x}{\y} : \V$ for all $\x, \y : \RezkH{\C}$.
To do so, we use the elimination rule,
and thus it suffices to construct
\begin{itemize}
  \item for all $\x, \y : \C$ an object $v(\x, \y) : \V$;
  \item for all isomorphisms $\f : \iso{\x_1}{\x_2}$ an identity $\app{l}{(\f , \y)} : v(\x_1, \y) = v(\x_2, \y)$;
  \item for all isomorphisms $\g : \iso{\y_1}{\y_2}$ an identity $\app{r}{(\x , \g)} : v(\x, \y_1) = v(\x, \y_2)$
\end{itemize}
such that the following equations hold.
\[
\app{l}{(\id{\x}, \y)} = \idpath,
\quad \quad
\app{r}{(\x , \id{\y})} = \idpath,
\]
\[
\app{l}{(\f \cdot \g, \y)} = \app{l}{(\f, \y)} \cdot \app{l}{(\g, \y)},
\quad \quad
\app{r}{(\x , \f \cdot \g)} = \app{r}{(\x , \f)} \cdot \app{r}{(\x , \g)},
\]
\[
\app{l}{(\f , \y_1)} \cdot \app{r}{(\x_2 , \g)}
=
\app{r}{(\x_1 , \g)} \cdot \app{l}{(\f , \y_2)}
\]
For $v(\x, \y)$, we take $\Ehom{\Ec}{\x}{\y}$.
To construct $\f$ and $\g$,
we use that $\V$ is univalent,
meaning that it suffices to construct isomorphisms
$\iso{v(\x_1, \y)}{v(\x_2, \y)}$
and
$\iso{v(\x, \y_1)}{v(\x, \y_2)}$
for all $\f : \iso{\x_1}{\x_2}$ and $\g : \iso{\y_1}{\y_2}$.
For these we take
$\EPrecomp{\f}$ and $\EPostcomp{\g}$ respectively.

Next we show how to construct the enriched identity,
for which we again use the elimination rule.
It suffices to construct for each $\x : \C$ a morphism $i : \Ehom{\Ec}{\x}{\x}$
such that for each isomorphism $f : \iso{\x_1}{\x_2}$ the following diagram commutes.
\[
\begin{tikzcd}[column sep = 4em]
  {\munit[V]} & {\Ehom{\Ec}{\x_1}{\x_1}} & {\Ehom{\Ec}{\x_2}{\x_1}} \\
  & {\Ehom{\Ec}{\x_2}{\x_2}}
  \arrow["{\app{i}{\x_1}}", from=1-1, to=1-2]
  \arrow["{\app{i}{\x_2}}"', from=1-1, to=2-2]
  \arrow["{\EPrecomp{(\f^{-1})}}", from=1-2, to=1-3]
  \arrow["{\EPostcomp{\f}}", from=1-3, to=2-2]
\end{tikzcd}
\]
The desired morphism $\app{i}{\x}$ is defined to be $\EidImpl{\x}$.
\end{construction}

Note that in \Cref{constr:enriched-rezk-completion-HIT}
we only use that $\V$ is a monoidal category,
and we do not need $\V$ to be complete or symmetric monoidal closed.
This is in contrast to \Cref{constr:enriched-rezk-completion} where we use closedness of the monoidal structure and completeness
to formulate and prove the Yoneda lemma (\Cref{prop:yoneda-lemma}).

Now let us reconsider the unit enriched category and the tensor of enriched categories (\Cref{exa:unit-enriched,exa:tensor-enriched}).
We already argued that both these examples did not necessarily give rise to a univalent enriched category,
so we can take their Rezk completion to obtain univalent analogues.
However, it crucially matters whether we use \Cref{constr:enriched-rezk-completion} or \Cref{constr:enriched-rezk-completion-HIT} for the Rezk completion.

More specifically, suppose that we want to construct the symmetric monoidal bicategory of univalent enriched categories in some universe $\UnivU$ \cite{kelly1982basic}.
The unit and tensor of this monoidal structure are given by the Rezk completions \Cref{exa:unit-enriched,exa:tensor-enriched} respectively.
If we were to use \Cref{constr:enriched-rezk-completion},
then it is not guaranteed that our bicategory of enriched categories is closed under the tensor.
This is because the universe level of $\Rezk{\mult{\Ec_1}{\Ec_2}}$ generally increases,
and thus this enriched category does not necessarily live in the desired bicategory.

However, if we assume that our universe is closed under higher inductive types,
then we can actually construct the desired monoidal structure.
Given univalent enriched categories $\Ec_1$ and $\Ec_2$ in some universe $\UnivU$,
both $\mult{\Ec_1}{\Ec_2}$ and $\RezkH{\mult{\Ec_1}{\Ec_2}}$ also live in universe $\UnivU$.
We can say the same for $\RezkH{\UnitE}$, which allows us to show that $\EnrichCat{\V}$ is monoidal.

\subsection{The Universal Property}
Next we prove the universal property of the Rezk completion (\cref{thm:rezk-completion-ump}),
so we assume that we have a fully faithful and essentially surjective enriched functor $\Ef : \Ec_1 \onecell \Ec_2$
and a univalent enriched category $\Ec_3$.
To verify \cref{thm:rezk-completion-ump}, we use that $\Ec_3$ is univalent.
This implies that the categories $\EFunctor{\Ec_2}{\Ec_3}$ and $\EFunctor{\Ec_1}{\Ec_3}$ are both univalent,
and thus it suffices to check that $\precomp{\Ef} : \EFunctor{\Ec_2}{\Ec_3} \onecell \EFunctor{\Ec_1}{\Ec_3}$ is essentially surjective and fully faithful.
The proofs of \cref{thm:rezk-completion-ump-ff,thm:rezk-completion-ump-eso} have some overlap with the ordinary categorical case \cite[Theorem 8.4]{rezk_completion}.
However, here we must also check that the obtained functors and natural transformations actually are enriched.

\begin{lemL}[\coqdocurl{CategoryTheory.EnrichedCats.RezkCompletion.PrecompFullyFaithful}{enriched_rezk_completion_ump_fully_faithful}]
\label{thm:rezk-completion-ump-ff}
The functor $\precomp{\Ef} : \EFunctor{\Ec_2}{\Ec_3} \onecell \EFunctor{\Ec_1}{\Ec_3}$ is fully faithful.
\end{lemL}

\begin{proof}
The proof that $\precomp{\Ef} : \EFunctor{\Ec_2}{\Ec_3} \onecell \EFunctor{\Ec_1}{\Ec_3}$ is faithful,
is in essence the same as for ordinary categories \cite[Lemma 8.1]{rezk_completion},
so we only show that $\precomp{\Ef}$ is full.
Let $\Eg_1, \Eg_2 : \EFunctor{\Ec_2}{\Ec_3}$ be two enriched functors,
and suppose that we have an enriched transformation $\nt : \comp{\Ef}{\Eg_1} \twocell \comp{\Ef}{\Eg_2}$.
We show how to construct the desired enriched natural transformation $\theta : \Eg_1 \twocell \Eg_2$.

For all objects $\x : \Ec_2$ the following type is contractible.
\[
\sigmatype{\f}{\app{\Eg_1}{\x} \onecell \app{\Eg_2}{\x}}{\prod (\w : \Ec_1) (i : \iso{\app{\Ef}{\w}}{\x}), \app{\nt}{\w} \cdot \app{\Eg_2}{i} = \app{\Eg_1}{i} \cdot \f}
\]
The contractibility of this type follows from our assumption that $\Ef$ is essentially surjective.
From this, we obtain the data of the the desired transformation $\theta$.
The fact that $\theta$ is $\V$-enriched is shown by using \cref{eq:nat-trans-enrichment} and the fact that $\Ef$ is essentially surjective.
\end{proof}

\begin{lemL}[\coqdocurl{CategoryTheory.EnrichedCats.RezkCompletion.PrecompEssentiallySurjective}{enriched_rezk_completion_ump_essentially_surjective}]
\label{thm:rezk-completion-ump-eso}
The functor $\precomp{\Ef} : \EFunctor{\Ec_2}{\Ec_3} \onecell \EFunctor{\Ec_1}{\Ec_3}$ is essentially surjective.
\end{lemL}

\begin{proof}
Suppose that we have an enriched functor $\Eg : \EFunctor{\Ec_1}{\Ec_3}$.
We only demonstrate how to construct the desired enriched functor $\Eh : \EFunctor{\Ec_2}{\Ec_3}$.

Suppose that we have $\x : \Ec_2$.
Then there is a unique object $\y : \Ec_3$ and function $\varphi : \prod (w : \Ec_1) (i : \iso{\app{\Ef}{\x}}{\y}), \iso{\app{\Eg}{\w}}{\y}$
such that for all objects $\w_1, \w_2 : \Ec_1$,
isomorphisms $i_1 : \iso{\app{\Ef}{\w_1}}{\x}$ and $i_2 : \iso{\app{\Ef}{\w_2}}{\x}$,
and morphisms $k : \w_1 \onecell \w_2$ satisfying $\app{\Ef}{k} \cdot i_2 = i_1$,
we have $\app{\Eg}{k} \cdot \app{\app{\varphi}{\w_2}}{i_2} = \app{\app{\varphi}{\w_1}}{i_1}$.
Uniqueness follows from the fact that $\Ef$ is fully faithful,
and the desired element is constructed by using that $\Ef$ is essentially surjective.
One can show that the obtained action on objects gives rise to a functor $H$ from the underlying category of $\Ec_2$ to that of $\Ec_3$.
We also have isomorphisms $\app{\app{\varphi}{\w}}{i} : \iso{\app{\Eg}{\w}}{\app{H}{\x}}$ for all $w : \Ec_1$ and $i : \iso{\app{\Ef}{\x}}{\y}$.

Next we construct an enrichment for this functor.
Suppose, that we have two objects $\x, \y : \Ec_2$.
Then there is a unique morphism $f : \Ehom{\Ec_2}{\x}{\y} \onecell \Ehom{\Ec_3}{\app{H}{\x}}{\app{H}{\y}}$ in $\V$
such that for all objects $w_1, w_2 : \Ec_2$
and isomorphisms $i_1 : \iso{\app{\Ef}{\w_1}}{\x}$ and $i_2 : \iso{\app{\Ef}{\w_2}}{\x}$,
$f$ is equal to the following composition of morphisms
\[
  \begin{tikzcd}
    {\Ehom{\Ec_2}{\x}{\y}} &[4em] {\Ehom{\Ec_2}{\app{\Ef}{\w_1}}{\y}} &[1em] {\Ehom{\Ec_2}{\app{\Ef}{\w_1}}{\app{\Ef}{\w_2}}} & {\Ehom{\Ec_1}{\w_1}{\w_2}} \\
    {\Ehom{\Ec_3}{\app{\Eg}{\w_1}}{\app{\Eg}{\w_2}}} & {\Ehom{\Ec_3}{\app{H}{x}}{\app{\Eg}{\w_2}}} & {\Ehom{\Ec_3}{\app{H}{x}}{\app{H}{y}}}
    \arrow["{\EPrecomp{i_1}}", from=1-1, to=1-2]
    \arrow["{\EPrecomp{(i_2^{-1})}}", from=1-2, to=1-3]
    \arrow["{(\Efun{\Ef}{\w_1}{\w_2})^{-1}}", from=1-3, to=1-4]
    \arrow["{\Efun{\Eg}{\w_1}{\w_2}}"{description}, from=1-4, to=2-1]
    \arrow["{\EPrecomp{((\app{\app{\varphi}{\w_1}}{i_1})^{-1})}}"', from=2-1, to=2-2]
    \arrow["{\EPostcomp{\app{\app{\varphi}{\w_2}}{i_2}}}"', from=2-2, to=2-3]
  \end{tikzcd}
\]
This follows from the fact that $\Ef$ is a weak equivalence.
As such, we get the desired enriched functor $\Eh : \EFunctor{\Ec_2}{\Ec_3}$.
\end{proof}

Let us finish this section by recalling \emph{enriched profunctors}.
Enriched profunctors are defined to be enriched functors $\mult{\OpE{\Ec_2}}{\Ec_1} \onecell \self{\V}$.
As discussed before, $\mult{\OpE{\Ec_2}}{\Ec_1}$ is not necessarily univalent.
However,
we can still define enriched profunctors using only univalent enriched categories.
If $\V$ is univalent,
then enriched functors $\mult{\OpE{\Ec_2}}{\Ec_1} \onecell \self{\V}$
correspond to enriched functors $\Rezk{\mult{\OpE{\Ec_2}}{\Ec_1}} \onecell \self{\V}$ by the universal property of the Rezk completion,
as indicated in diagram below.
\[
\begin{tikzcd}
  {\mult{\OpE{\Ec_2}}{\Ec_1}} & {\self{\V}} \\
  {\Rezk{\mult{\OpE{\Ec_2}}{\Ec_1}}}
  \arrow[""{name=0, anchor=center, inner sep=0}, "\Ef", from=1-1, to=1-2]
  \arrow["\RezkFun"', from=1-1, to=2-1]
  \arrow[""{name=1, anchor=center, inner sep=0}, "\Eg"', dashed, from=2-1, to=1-2]
  \arrow[shorten <=3pt, shorten >=3pt, Rightarrow, from=0, to=1]
\end{tikzcd}
\]

\section{Enriched Monads}
\label{sec:enriched-monads}
We end our study of univalent enriched categories by looking at enriched monads.
More specifically, we discuss Kleisli objects (\cref{constr:enriched-kleisli-cat}) in the bicategory of enriched categories.
At first glance, it might not seem that univalence plays an interesting role,
but upon closer look, this question is rather subtle.

Usually, the Kleisli category of a monad $\T$ on a category $\C$ is defined to be the category
whose objects are objects of $\C$ and whose morphisms from $\x$ to $\y$ are morphisms $\x \onecell \app{\T}{\y}$ in $\C$.
We denote this category by $\FKleisli{\T}$.
In general, this category is not univalent (for example the constant monad on the unit set).
This situation can be rectified by defining the Kleisli category in a slightly different way \cite{univalence-principle},
namely as the image of the free algebra functor from $\C$ to the Eilenberg-Moore category $\EM{\T}$ of $\T$.
The resulting univalent category is denoted by $\Kleisli{\T}$.
To derive the usual theorems about Kleisli categories, one can instantiate the formal theory of monads \cite{MR1935981,street1972formal,Weide23},
meaning that it suffices to prove the universal property for Kleisli objects.
Proving the desired universal property is a nice exercise using the universal property of the Rezk completion (\cref{thm:rezk-completion-ump}).

The key notion of this section, enriched monads, can be defined concisely as monads internal to $\EnrichCat{\V}$.
Recall that monads in bicategories are defined as follows.

\begin{defi}
\label{def:monad}
Let $\B$ be a bicategory.
A \conceptDef{monad $m$ in $\B$}{Bicategories.DisplayedBicats.Examples.MonadsLax}{mnd} is given by
\begin{itemize}
  \item an object $\monadob{m} : \B$;
  \item a 1-cell $\monadendofull{m} : \monadob{m} \onecell \monadob{m}$;
  \item a 2-cell $\monadunit{m} : \id[\monadob{m}] \twocell \monadendo{m}$;
  \item a 2-cell $\monadmult{m} : \comp{\monadendo{m}}{\monadendo{m}} \twocell \monadendo{m}$.
\end{itemize}
such that the following diagrams commute.
\[
\begin{tikzcd}
  {\monadendofull{m}} & {\comp{\monadendofull{m}}{\id{x}}} &[3em] {\comp{\monadendofull{m}}{\monadendofull{m}}} &[3em] {\comp{\id}{\monadendofull{m}}} & {\monadendofull{m}} \\
  && {\monadendofull{m}}
  \arrow["{\monadmult{m}}", from=1-3, to=2-3, Rightarrow]
  \arrow["{\monadendofull{m} \whiskerl \monadunit{m}}", from=1-2, to=1-3, Rightarrow]
  \arrow["{\rinvunitor{}}", from=1-1, to=1-2, Rightarrow]
  \arrow["{\id}"', from=1-1, to=2-3, Rightarrow]
  \arrow["{\linvunitor{}}"', from=1-5, to=1-4, Rightarrow]
  \arrow["{\monadunit{m} \whiskerr \monadendofull{m}}"', from=1-4, to=1-3, Rightarrow]
  \arrow["{\id}", from=1-5, to=2-3, Rightarrow]
\end{tikzcd}
\]
\[
\begin{tikzcd}
  {\comp{\monadendofull{m}}{(\comp{\monadendofull{m}}{\monadendofull{m}})}} &[5em] & {\comp{\monadendofull{m}}{\monadendofull{m}}} \\
  {\comp{(\comp{\monadendofull{m}}{\monadendofull{m}})}{\monadendofull{m}}} & {\comp{\monadendofull{m}}{\monadendofull{m}}} & {\monadendofull{m}}
  \arrow["{\monadmult{m}}", from=1-3, to=2-3, Rightarrow]
  \arrow["{\monadmult{m}}"', from=2-2, to=2-3, Rightarrow]
  \arrow["{\monadmult{m} \whiskerr \monadendofull{m}}"', from=2-1, to=2-2, Rightarrow]
  \arrow["{\lassociator{}{}{}}"', from=1-1, to=2-1, Rightarrow]
  \arrow["{\monadendofull{m} \whiskerl \monadmult{m}}", from=1-1, to=1-3, Rightarrow]
\end{tikzcd}
\]
Here $\lunitor{}$ and $\runitor{}$ are the left and right unitors of $\B$, and $\lassociator{}{}{}$ is the associator of $\B$.
\end{defi}

For enriched categories, one can further unfold this definition and phrase it in terms of enrichments.
This results in the notion of \emph{enrichments for monads}.

\begin{defi}
\label{def:enriched-monad}
Suppose that we have a category $\C$, a monad $\T$ on $\C$, and a $\V$-enrichment for $\C$.
Then a \conceptDef{$\V$-enrichment}{CategoryTheory.EnrichedCats.EnrichmentMonad}{monad_enrichment} for $\T$ consists of a $\V$-enrichment for the endofunctor $\T$
such that the unit $\unitM{\T}$ and multiplication $\muM{\T}$ are $\V$-enriched natural transformations.
\end{defi}

When we say \emph{enriched monad}, we mean a monad together with an enrichment.
In the remainder of this section,
we are concerned with \emph{Kleisli objects} in the bicategory of enriched categories.
To define Kleisli objects, we first define their cocones.
Note that for these definitions, we talk about arbitrary bicategories $\B$ and monads internal to $\B$.

\begin{defi}
Let $\B$ be a bicategory and let $m$ be a monad in $\B$.
A \conceptDef{Kleisli cocone}{Bicategories.Colimits.KleisliObjects}{kleisli_cocone} $k$ for $m$ in $\B$ consists of an object $\klob{k} : \B$, a 1-cell $\klmor{k} : \monadob{m} \onecell \klob{k}$, and a 2-cell $\klcell{k} : \comp{\monadendo{m}}{\klmor{k}} \twocell \klmor{k}$ such that the following diagrams commute.
\[
\begin{tikzcd}[column sep = huge]
  {\comp{\id[\monadob{m}]}{\klmor{k}}} & {\comp{\monadendo{m}}{\klmor{k}}} \\
  & {\klmor{k}}
  \arrow["{\monadunit{m} \whiskerr \klmor{k}}", from=1-1, to=1-2, Rightarrow]
  \arrow["{\klcell{k}}", from=1-2, to=2-2, Rightarrow]
  \arrow["{\lunitor{\klmor{k}}}"', from=1-1, to=2-2, Rightarrow]
\end{tikzcd}
\]
\[
\begin{tikzcd}[column sep = large]
  {\comp{(\comp{\monadendo{m}}{\monadendo{m}})}{\klmor{k}}} & {\comp{\monadendo{m}}{(\comp{\monadendo{m}}{\klmor{k}})}} & {\comp{\monadendo{m}}{\klmor{k}}} \\
  {\comp{\monadendo{m}}{\klmor{k}}} && {\klmor{k}}
  \arrow["{\monadmult{m} \whiskerr \klmor{k}}"', from=1-1, to=2-1, Rightarrow]
  \arrow["{\rassociator{\monadendo{m}}{\monadendo{m}}{\klmor{k}}}", from=1-1, to=1-2, Rightarrow]
  \arrow["{\monadendo{m} \whiskerl \klcell{k}}", from=1-2, to=1-3, Rightarrow]
  \arrow["{\klcell{k}}", from=1-3, to=2-3, Rightarrow]
  \arrow["{\klcell{k}}"', from=2-1, to=2-3, Rightarrow]
\end{tikzcd}
\]
\end{defi}

\begin{defi}
A Kleisli cocone $k$ is \conceptDef{universal}{Bicategories.Colimits.KleisliObjects}{has_kleisli_ump} if the following conditions are satisfied.
\begin{itemize}
\item For every Kleisli cocone $q$ there is a 1-cell $\klumpmor{q} : \klob{k} \onecell \klob{q}$ and an invertible 2-cell $\klumpcom{q} : \comp{\klmor{k}}{\klumpmor{q}} \twocell \klmor{q}$ such that the following diagram commutes.
  \[
    \begin{tikzcd}[column sep = huge]
      {\comp{\monadendo{m}}{(\comp{\klmor{k}}{\klumpmor{q}})}} && {\comp{\monadendo{m}}{\klmor{q}}} \\
      {\comp{(\comp{\monadendo{m}}{\klmor{k}})}{\klumpmor{q}}} & {\comp{\klmor{k}}{\klumpmor{q}}} & {\klmor{q}}
      \arrow["{\monadendo{m} \whiskerl \klumpcom{q}}", from=1-1, to=1-3, Rightarrow]
      \arrow["{\lassociator{\monadendo{m}}{\klmor{k}}{\klumpmor{q}}}"', from=1-1, to=2-1, Rightarrow]
      \arrow["{\klcell{k} \whiskerr \klumpmor{q}}"', from=2-1, to=2-2, Rightarrow]
      \arrow["{\klumpcom{q}}"', from=2-2, to=2-3, Rightarrow]
      \arrow["{\klcell{q}}", from=1-3, to=2-3, Rightarrow]
    \end{tikzcd}
  \]
\item Suppose that we have an object $x : \B$, two 1-cells $g_1, g_2 : \klob{k} \onecell x$, and a 2-cell $\tau : \comp{\klmor{k}}{g_1} \twocell \comp{\klmor{k}}{g_2}$ such that the following diagram commutes.
  \[
    \begin{tikzcd}[column sep = huge]
      {\comp{\monadendo{m}}{(\comp{\klmor{k}}{g_1})}} & {\comp{(\comp{\monadendo{m}}{\klmor{k}})}{g_1}} & {\comp{\klmor{k}}{g_1}} \\
      {\comp{\monadendo{m}}{(\comp{\klmor{k}}{g_2})}} & {\comp{(\comp{\monadendo{m}}{\klmor{k}})}{g_2}} & {\comp{\klmor{k}}{g_2}}
      \arrow["{\lassociator{\monadendo{m}}{\klmor{k}}{g_1}}", from=1-1, to=1-2, Rightarrow]
      \arrow["{\klcell{k} \whiskerr g_1}", from=1-2, to=1-3, Rightarrow]
      \arrow["\tau", from=1-3, to=2-3, Rightarrow]
      \arrow["{\monadendo{m} \whiskerl \tau}"', from=1-1, to=2-1, Rightarrow]
      \arrow["{\lassociator{\monadendo{m}}{\klmor{k}}{g_2}}"', from=2-1, to=2-2, Rightarrow]
      \arrow["{\klcell{k} \whiskerr g_2}"', from=2-2, to=2-3, Rightarrow]
    \end{tikzcd}
  \]
  Then there is a unique 2-cell $\klumpcell{\tau} : g_1 \twocell g_2$ such that $\klcell{k} \whiskerl \klumpcell{\tau} = \tau$.
\end{itemize}
We say that a bicategory \conceptDef{has Kleisli objects}{Bicategories.Colimits.KleisliObjects}{has_kleisli} if there is a universal Kleisli cocone for every monad $m$.
\end{defi}

As discussed before,
there are multiple ways to define Kleisli categories.
We first define an enrichment for $\FKleisli{\T}$.

\begin{exa}
\label{exa:kleisli-cat-enrichment}
Let $\T$ be an enriched monad on an enriched category $\Ec$.
We define a \conceptDef{$\V$-enrichment $\FKleisliE{\T}$ for $\FKleisli{\T}$}{CategoryTheory.EnrichedCats.Examples.KleisliEnriched}{Kleisli_cat_monad_enrichment} as follows.
\begin{itemize}
  \item We define $\Ehom{\FKleisliE{\T}}{\x}{\y}$ to be $\Ehom{\Ec}{\x}{\app{\T}{\y}}$.
  \item We define $\EidImpl{\x}$ to be $\EFromArr{\app{\unitM{\T}}{\x}}$.
  \item We define $\EcompImpl{\x}{\y}{\z}$ as the following composition of morphisms.
    \[\begin{tikzcd}
	{\mult{\Ehom{\Ec}{\y}{\app{\T}{\z}}}{\Ehom{\Ec}{\x}{\app{\T}{\y}}}} & {\mult{\Ehom{\Ec}{\app{\T}{\y}}{\app{\T}{(\app{\T}{\z})}}}{\Ehom{\Ec}{\x}{\app{\T}{\y}}}} &[-4pt] {\Ehom{\Ec}{\x}{\app{\T}{(\app{\T}{\z})}}} & {\Ehom{\Ec}{\x}{\app{\T}{\z}}}
	\arrow["{\Ecomp{\x}{\app{\T}{\y}}{\app{\T}{(\app{\T}{\z})}}}", from=1-2, to=1-3]
	\arrow["{\mult{\Efun{\T}{\y}{\app{\T}{\z}}}{\id{}}}", from=1-1, to=1-2]
	\arrow["{\EPostcomp{(\app{\muM{\T}}{\z})}}", from=1-3, to=1-4]
      \end{tikzcd}\]
\end{itemize}
The operations $\EFromArr{\f}$ and $\EToArr{\f}$ in $\FKleisliE{\T}$ are inherited from $\Ec$.
\end{exa}

Next we define an enrichment for $\Kleisli{\T}$.
Since $\Kleisli{\T}$ is defined as a full subcategory of the Eilenberg-Moore category $\EM{\T}$,
we define an enrichment for $\EM{\T}$ first.

\begin{exa}
\label{exa:em-cat-enrichment}
Suppose that $\V$ has equalizers, and let $\T$ be an enriched monad on an enriched category $\Ec$.
Note that we can define the Eilenberg-Moore category of $\T$ as a full subcategory of $\Dialg{\T}{\id}$.
By \cref{exa:dialgebras-enrichment,exa:full-sub-enrichment} we obtain the desired \conceptDef{$\V$-enrichment $\EME{\T}$ on $\EM{\T}$}{CategoryTheory.EnrichedCats.Examples.EilenbergMooreEnriched}{eilenberg_moore_enrichment}.
\end{exa}

Using \cref{exa:em-cat-enrichment} one can show that $\EnrichCat{\V}$ has Eilenberg-Moore objects.
In general, we have an enriched functor $\freealg{\T} : \Ec \onecell \EME{\T}$.
This functor sends every object $\x$ to the free algebra $\app{\T}{\x}$.
Now we define an enrichment for $\Kleisli{\T}$.

\begin{exa}
\label{exa:univ-kleisli-cat-enrichment}
Suppose that $\V$ is a monoidal category with equalizers,
and let $\T$ be an enriched monad on an enriched category $\Ec$.
Note that $\Kleisli{\T}$ is constructed as a full subcategory of the Eilenberg-Moore category,
and thus by \cref{exa:em-cat-enrichment} we obtain the \conceptDef{$\V$-enrichment $\KleisliE{\T}$ for $\Kleisli{\T}$}{CategoryTheory.EnrichedCats.Examples.UnivalentKleisliEnriched}{kleisli_cat_enrichment}.
\end{exa}

The category defined in \cref{exa:univ-kleisli-cat-enrichment} is univalent if we assume $\Ec$ to be univalent.
This is because the Eilenberg-Moore category of a monad on a univalent category is always univalent
and because univalence is preserved under full subcategories.
In addition, note that in \cref{exa:univ-kleisli-cat-enrichment} we assume that $\V$ has equalizers,
whereas in \cref{exa:kleisli-cat-enrichment}, we do not.

We finish this section by showing that $\KleisliE{\T}$ satisfies the required universal property.
The main idea behind the proof is that we have a weak equivalence $\kleislifunctor{\T} : \FKleisliE{\T} \onecell \KleisliE{\T}$,
and this weak equivalence allows use to instantiate \cref{thm:rezk-completion-ump}.

\begin{problem}
\label{prob:enriched-kleisli-cat}
Given a monoidal category $\V$ with equalizers, to construct Kleisli objects in the bicategory $\EnrichCat{\V}$.
\end{problem}

\begin{construction}{\coqdocurl{Bicategories.Colimits.Examples.BicatOfEnrichedCatsColimits}{bicat_of_enriched_cats_has_kleisli}}{prob:enriched-kleisli-cat}
\label{constr:enriched-kleisli-cat}
Given an enriched monad $\T$ on $\Ec$,
the Kleisli object of $\T$ in $\EnrichCat{\V}$ is given by $\KleisliE{\T}$.
The main work lies in verifying the universal property.
This check happens in three steps.

First, we define a weak equivalence $\kleislifunctor{\T} : \FKleisliE{\T} \onecell \KleisliE{\T}$.
This enriched functor sends every object $\x$ to the free algebra on $\x$.
The action on morphisms is given by the following composition
\[
  \begin{tikzcd}
    {\Ehom{\Ec}{\x}{\app{\T}{\y}}} & {\Ehom{\Ec}{\app{\T}{\x}}{\app{\T}{(\app{\T}{\y})}}} &[3em] {\Ehom{\Ec}{\app{\T}{\x}}{\app{\T}{\y}}}
    \arrow["{\Efun{\T}{\x}{\y}}", from=1-1, to=1-2]
    \arrow["{\EPostcomp{(\muM{\y})}}", from=1-2, to=1-3]
  \end{tikzcd}
\]

Second, we check that $\FKleisliE{\T}$ gives rise to Kleisli objects in the bicategory of (not necessarily univalent) enriched categories.
For this, one can use the same proof as used, for example, by Street \cite[Theorem 15]{street1972formal}.

Third, we conclude that the universal property also holds for $\KleisliE{\T}$,
and we only show how to construct 1-cells arising from the mapping property.
Suppose, that we have some Kleisli cocone $q$ in $\EnrichCat{\V}$.
We get an enriched functor $\Ef : \FKleisliE{\T} \onecell \klob{q}$.
From \cref{thm:rezk-completion-ump}, we get the desired 1-cell $\tilde{\Ef} : \KleisliE{\T} \onecell \klob{q}$.
\end{construction}

Note the similarities between \cref{constr:enriched-kleisli-cat} and the construction of Kleisli objects for univalent categories \cite[Construction 6.10]{Weide23}.


\section{Conclusion}
\label{sec:conclusion}
In this paper, we studied univalent enriched categories,
and we discussed several aspects of their study.
Our notion of univalent enriched category was based on enrichments,
and we viewed enriched categories as a category together with an enrichment.
First, we proved a structure identity principle for univalent enriched categories,
which we formulated using univalent bicategories.
The proof used displayed bicategories.
Second, we showed that all weak equivalences between univalent enriched categories
are adjoint equivalences.
Here we made use of orthogonal factorization systems.
Third, we discussed the Rezk completion of enriched categories,
which we constructed using the Yoneda lemma.
We also used the Rezk completion to construct Kleisli objects in the bicategory of univalent enriched categories.

Along the way, we saw a couple of interesting points where univalence interacted with enrichment.
When we defined the change-of-base operation in \cref{exa:change-of-base},
we restricted ourselves to lax functors that preserve underlying categories.
This was to guarantee that the resulting category would remain univalent.
In addition, we assumed that the monoidal category $\V$ has equalizers in the construction of the univalent Kleisli category (\cref{exa:univ-kleisli-cat-enrichment}).

There are several ways to extend the results in this paper.
A wide variety of notions in category theory can be defined internally to a bicategory.
However, for enriched categories, these internal notions are not always the correct ones.
For example, the notion of a fully faithful 1-cell can be defined internally to a bicategory using a representable definition,
but the obtained notion does not correspond to the one given in \cref{defi:ff-enriched}.
To obtain the desired notions, one could use the theory of equipments \cite{proarrows},
and one interesting extension of this work would be to develop the equipment of enriched categories.
Such work would build forth upon recent work on univalent double (bi)categories \cite{insights-paper,DBLP:journals/corr/abs-2310-09220,MR2844536}.
Another interesting extension would be formalizing applications of enriched categories, such as models of the enriched effect calculus \cite{EggerMS14} or enriched profunctor optics \cite{DBLP:journals/corr/abs-2001-07488}.
Finally,
our notion of structure that supports smash products is limited in its applicability,
due to the fact that we construct smash products using quotients of sets.
As mentioned in \Cref{rem:dcpo-smash-prod},
this notion is not suitable to construct smash products of DCPOs.
A question for future research would how to generalize this notion,
so that it encapsulates a wider range of examples.
Such a generalization would require one to use quotients in the appropriate category
instead of sets.

\section*{Acknowledgments}

The author thanks the anonymous reviewers of HoTT/UF, FSCD, and LMCS for their useful comments,
Nima Rasekh for useful discussions,
and Dan Frumin and Kobe Wullaert for proof reading earlier versions of this paper.
The author also thanks the Rocq developers for providing the Rocq proof assistant and their continuous support to keep \UniMath compatible with Rocq.
This research was supported by the NWO project “The Power of Equality” OCENW.M20.380, which is financed by the Dutch Research Council (NWO).

\bibliographystyle{alphaurl}
\bibliography{literature}

\end{document}